\documentclass[sigconf]{aamas} 
\usepackage{balance} % for balancing columns on the final page

\usepackage{url}
\usepackage[utf8]{inputenc}
\usepackage{graphicx}
\usepackage{amsmath}
\usepackage[ruled,vlined,linesnumbered]{algorithm2e}
\usepackage{color}
\usepackage{amsthm}
\usepackage{dsfont}
\usepackage[noend]{algpseudocode}
\usepackage{subcaption}
\usepackage{tabularx}
\usepackage{enumerate}
\usepackage{cases}
\usepackage{multicol}
\usepackage{multirow}
\usepackage{enumitem}
\usepackage{bm}
\usepackage{import}

\definecolor{linkBlue}{rgb}{0,0.43,0.72}

\newtheorem{innercustomthm}{Theorem}
\newenvironment{customthm}[1]
  {\renewcommand\theinnercustomthm{#1}\innercustomthm}
  {\endinnercustomthm}

\setcopyright{ifaamas}
\acmConference[AAMAS '21]{Proc.\@ of the 20th International Conference on Autonomous Agents and Multiagent Systems (AAMAS 2021)}{May 3--7, 2021}{Online}{U.~Endriss, A.~Now\'{e}, F.~Dignum, A.~Lomuscio (eds.)}
\copyrightyear{2021}
\acmYear{2021}
\acmDOI{}
\acmPrice{}
\acmISBN{}

\acmSubmissionID{198}

\title[Improved Cooperation by Exploiting a Common Signal]{Improved Cooperation by Exploiting a Common Signal}
\subtitle{Learning Temporal Conventions for Sustainable Appropriation of Common-Pool Resources}

\author{Panayiotis Danassis}
\affiliation{%
	\institution{\'Ecole Polytechnique F\'ed\'erale de Lausanne (EPFL) \\ Artificial Intelligence Laboratory}
	\city{Lausanne} 
	\country{Switzerland} 
	\postcode{CH-1015}
}
\email{panayiotis.danassis@epfl.ch}

\author{Zeki Doruk Erden}
\affiliation{%
	\institution{\'Ecole Polytechnique F\'ed\'erale de Lausanne (EPFL) \\ Artificial Intelligence Laboratory}
	\city{Lausanne} 
	\country{Switzerland} 
	\postcode{CH-1015}
}
\email{zeki.erden@epfl.ch}

\author{Boi Faltings}
\affiliation{%
	\institution{\'Ecole Polytechnique F\'ed\'erale de Lausanne (EPFL) \\ Artificial Intelligence Laboratory}
	\city{Lausanne} 
	\country{Switzerland} 
	\postcode{CH-1015}
}
\email{boi.faltings@epfl.ch}

\begin{abstract}
	Can artificial agents benefit from human conventions? Human societies manage to successfully self-organize and resolve the tragedy of the commons in common-pool resources, in spite of the bleak prediction of non-cooperative game theory. On top of that, real-world problems are inherently large-scale and of low observability. One key concept that facilitates human coordination in such settings is the use of conventions. Inspired by human behavior, we investigate the learning dynamics and emergence of temporal conventions, focusing on common-pool resources. Extra emphasis was given in designing a \emph{realistic evaluation setting}: (a) environment dynamics are modeled on real-world fisheries, (b) we assume decentralized learning, where agents can observe only their own history, and (c) we run large-scale simulations (up to 64 agents).

	Uncoupled policies and low observability make cooperation hard to achieve; as the number of agents grow, the probability of taking a correct gradient direction decreases exponentially. By introducing an \emph{arbitrary common signal} (e.g., date, time, or any periodic set of numbers) as a means to couple the learning process, we show that temporal conventions can emerge and agents reach \emph{sustainable} harvesting strategies. The introduction of the signal consistently improves the social welfare (by $258\%$ on average, up to $3306\%$), the range of environmental parameters where sustainability can be achieved (by $46\%$ on average, up to $300\%$), and the convergence speed in low abundance settings (by $13\%$ on average, up to $53\%$).
\end{abstract}


\newcommand{\BibTeX}{\rm B\kern-.05em{\sc i\kern-.025em b}\kern-.08em\TeX}

\begin{document}

%%% The following commands remove the headers in your paper. For final 
%%% papers, these will be inserted during the pagination process.

\pagestyle{fancy}
\fancyhead{}

%%% The next command prints the information defined in the preamble.

\maketitle 

%%%%%%%%%%%%%%%%%%%%%%%%%%%%%%%%%%%%%%%%%%%%%%%%%%%%%%%%%%%%%%%%%%%%%%%%

\section{Introduction}

The question of \emph{cooperation} in socio-ecological systems and \emph{sustainability} in the use of common-pool resources constitutes a critical open problem. Classical non-cooperative game theory suggests that rational individuals will exhaust a common resource, rather than sustain it for the benefit of the group, resulting in the `the tragedy of the commons' \cite{hardin1968tragedy}. The tragedy of the commons arises when it is challenging and/or costly to exclude individuals from appropriating common-pool resources (CPR) of finite yield \cite{ostrom1994rules}. Individuals face strong incentives to appropriate, which results in \emph{overuse} and even \emph{permanent depletion} of the resource. Examples include the degradation of fresh water resources, the over-harvesting of timber, the depletion of grazing pastures, the destruction of fisheries, etc.

In spite of the bleak prediction of non-cooperative game theory, the tragedy of the commons is not inevitable, though conditions under which cooperation and sustainability can be achieved may be more demanding, the higher the stakes. Nevertheless, humans have been systematically shown to successfully self-organize and resolve the tragedy of the commons in CPR appropriation problems, even without the imposition of an extrinsic incentive structure \cite{ostrom1999coping}. E.g., by enabling the capacity to communicate, individuals have been shown to maintain the harvest to an optimal level \cite{ostrom1994rules,casari2003decentralized}. Though, communication creates overhead, and might not always be possible \cite{AAAI10-adhoc}. One of the key findings of empirical field research on sustainable CPR regimes around the world is the employment of \emph{boundary rules}, which prescribe who is authorized to appropriate from a resource \cite{ostrom1999coping}. Such boundary rules can be of temporal nature, prescribing the \emph{temporal order} in which people harvest from a common-pool resource (e.g., `protocol of play' \cite{BUDESCU1997179}). The aforementioned rules can be enforced by an authority, or emerge in a self-organized manner (e.g., by utilizing environmental signals such as the time, date, season, etc.) in the form of a \emph{social convention}.
% (due to for example partial observability, different communication protocols, existence of legacy agents, etc.)
% as autonomous agents proliferate and their lifespan increases, communication might not be possible

Many real-world CPR problems are inherently \emph{large-scale} and \emph{partially observable}, which further increases the challenge of sustainability. In this work we deal with the \emph{most information-restrictive setting}: each participant is modeled as an individual agent with its own policy conditioned only on \emph{local information}, specifically his own history of action/reward pairs (\emph{fully decentralized} method). Global observations, including the resource stock, the number of participants, and the joint observations and actions, are hidden -- as is the case in many real-world applications, like commercial fisheries. Under such a setting, it is \emph{impossible to avoid positive probability mass on undesirable actions} (i.e., simultaneous appropriation), since there is no correlation between the agents' policies. This leads to either low social welfare, because the agents are being conservative, or, even worse, the \emph{depletion} of the resource. Depletion becomes more likely as the problem size grows due to the \emph{non-stationarity} of the environment and the \emph{global exploration problem}.

We propose a simple technique: allow agents to observe an \emph{arbitrary, common signal} from the environment. Observing a common signal mitigates the aforementioned problems because it allows for \emph{coupling} between the learned policies, increasing the joint policy space. Agents, for example, can now learn to harvest in turns, and with varying efforts per signal value, or allow for fallow periods. The benefit is twofold: the agents learn to not only avoid depletion, but also to maintain a healthy stock which allows for large harvest and, thus, higher social welfare. It is important to stress that \emph{we do not assume any a priori relation between the signal space and the problem at hand}. Moreover, we require no communication, no extrinsic incentive mechanism, and we do not change the underlying architecture, or learning algorithm. We simply utilize a means -- common environmental signals that are \emph{amply available to the agents} \cite{hart2000simple} -- to accommodate correlation between policies. This in turn enables the emergence of \emph{ordering conventions} of temporal nature (henceforth referred to as temporal conventions) and \emph{sustainable harvesting} strategies.

% hard to avoid catastrophic scenario 

\subsection{Our Contributions} \label{Our Contributions}

\noindent
\textbf{(1) We are the first to introduce a realistic common-pool resource appropriation game for multi-agent coordination}, based on bio-economic models of commercial fisheries, and provide theoretical analysis on the dynamics of the environment.
% Specifically, we prove the optimal harvesting strategy that maximizes the revenue, and identify two interesting stock values: the `limit of sustainable harvesting', and the `limit of immediate depletion'.

\noindent
\textbf{(2) We propose a simple and novel technique: allow agents to observe an arbitrary periodic environmental signal.} Such signals are \emph{amply available} in the environment (e.g., time, date etc.) and can \emph{foster cooperation} among agents.

\noindent
\textbf{(3) We provide a thorough (quantitative \& qualitative) analysis} on the learned policies and demonstrate significant improvements on sustainability, social welfare, and convergence speed.
% A qualitative analysis of the emerged policies and show the emergence of temporal harvesting conventions.

\subsection{Discussion \& Related Work} \label{Related Work}

As autonomous agents proliferate, they will be called upon to interact in ever more complex environments. This will bring forth the need for techniques that enable the emergence of sustainable cooperation. Despite the growing interest in and success of multi-agent deep reinforcement learning (MADRL), scaling to environments with a large number of learning agents continues to be a problem \cite{10.5555/3398761.3398813}. A multi-agent setting is inherently susceptible to many pitfalls: non-stationarity (moving-target problem), curse of dimensionality, credit assignment, global exploration, relative overgeneralization \cite{hernandez2019survey,matignon_laurent_le_fort_piat_2012,10.5555/997339}\footnote{Some of these adversities can be mitigated by the centralized training, decentralized execution paradigm. Yet, centralized methods likewise suffer from a plethora of other problems: they are computationally heavy, assume unlimited communication (which is impractical in many real-world applications), the exact same team has to be deployed (in the real-world we cooperate with strangers), and, most importantly, the size of the joint action space grows exponentially with the number of agents.}. Recent advances in the field of MADRL deal with only a limited number of agents. It is shown that as the number of agents increase, the probability of taking a correct gradient direction decreases exponentially \cite{hernandez2019survey}, thus the proposed methods cannot be easily generalized to complex scenarios with many agents.
% relative overgeneralization, which occurs when agents gravitate towards a robust but sub-optimal joint policy due to noise induced by the mutual influence of each agent’s exploration strategy on others’ learning updates.

Our approach aims to mitigate the aforementioned problems of MADRL by introducing \emph{coupling} between the learned policies. It is important to note that the proposed approach does not change the underlying architecture of the network (the capacity of the network stays the same), nor the learning algorithm or the reward structure. We simply augment the input space by allowing the observation of an arbitrary common signal. The signal has no a priori relation to the problem, i.e., \emph{we do not need to design an additional feature}; in fact \emph{we use a periodic sequence of arbitrary integers}. It is still possible for the original network (without the signal) to learn a sustainable strategy. Nevertheless, we show that the simple act of augmenting the input space drastically increases the social welfare, speed of convergence, and the range of environmental parameters in which sustainability can be achieved. Most importantly, the proposed approach requires no communication, creates no additional overhead, it is simple to implement, and scalable.

The proposed technique was inspired by temporal conventions in resource allocation games of non-cooperative game theory. The closest analogue is the courtesy convention of \cite{Danassis:2019:CMC:3306127.3331754}, where rational agents learn to coordinate their actions to access a set of indivisible resources by observing a signal from the environment. Closely related is the concept of the correlated equilibrium (CE) \cite{AUMANN197467,nisan2007algorithmic}, which, from a practical perspective, constitutes perhaps the most relevant non-cooperative solution concept \cite{hart2000simple}\footnote{Correlated equilibria also relate to boundary rules and temporal conventions in human societies; the most prominent example of a CE in real life is the traffic lights, which can also be viewed as a temporal convention for the use of the road.}. Most importantly, it is possible to achieve a correlated equilibrium without a central authority, simply by utilizing meaningless environmental signals \cite{Danassis:2019:CMC:3306127.3331754,cigler2013decentralized,borowski2014learning}. Such common environmental signals are \emph{amply available to the agents} \cite{hart2000simple}. The aforementioned line of research studies pre-determined strategies of rational agents. Instead, we study the emergent behaviors of a group of independent learning agents aiming to maximize the long term discounted reward. 
%(can also be viewed as a model of bounded rationality)
%must learn policies to implement their strategic decisions, and must do so while coping with the non-stationarity arising from other agents learning simultaneously.

%The aforementioned adversities are present in some degree in problems faced by humans.
%Behavioral conventions are a fundamental part of human societies, yet they have not appeared meaningfully in empirical modeling of multi-agent systems.
A second source of inspiration is behavioral conventions; one of the key concepts that facilitates human coordination\footnote{Humans are able to routinely and robustly cooperate in their every day lives in large-scale and under dynamic and unpredictable demand. They also have access to auxiliary information that help correlated their actions (e.g., time, date etc.).}. A convention is defined as a customary, expected, and self-enforcing behavioral pattern \cite{young1996economics,lewis2008convention}. It can be considered as a behavioral rule, designed and agreed upon ahead of time \cite{SHOHAM1995231,walker95understandingThe}, or it may emerge from within the system itself \cite{mihaylov2014decentralized,walker95understandingThe}. The examined temporal convention in this work falls on the latter category.

Moving on to the application domain, there has been great interest recently in CPR problems (and more generally, social dilemmas \cite{kollock1998social}) as an application domain for MADRL \cite{perolat2017multi,leibo2017multi,peysakhovich2017consequentialist,hughes2018inequity,peysakhovich2018prosocial,pmlr-v97-jaques19a,10.5555/3306127.3331756,10.5555/3398761.3398855,10.5555/3398761.3399016}. CPR problems offer complex environment dynamics and relate to real-world socio-ecological systems. There are a few distinct differences between the CPR models presented in the aforementioned works and the model introduced in this paper: First and foremost, we designed our model to \emph{resemble reality as closely as possible using bio-economic models of commercial fisheries} \cite{clark2006worldwide,diekert2012tragedy}, resulting in complex environment dynamics. Second, we have a \emph{continuous action space} which further complicates the learning process. Finally, we opted not to learn from visual input (raw pixels). The problem of direct policy approximation from visual input does not add complexity to the social dilemma itself; it only adds complexity in the feature extraction of the state. It requires large networks because of the additional complexity of extracting features from pixels, while only a small part of what is learned is the actual policy \cite{10.5555/3306127.3331796}. Most importantly, it makes harder to study the policy in isolation, as we do in this work. Moreover, from a practical perspective, learning from a visual input would be meaningless, given that we are dealing with a low observability scenario where the resource stock and the number and actions of the participants are hidden.

In terms of the methodology for dealing with the tragedy of the commons, the majority of the aforementioned literature falls broadly into two categories: Reward shaping \cite{hughes2018inequity,pmlr-v97-jaques19a,peysakhovich2018prosocial}, which refers to adding a term to the extrinsic reward an agent receives from the environment, and opponent shaping \cite{10.5555/3398761.3399016,10.5555/3398761.3398855,perolat2017multi}, which refers to manipulating the opponent (by e.g., sharing rewards, punishments, or adapting your own actions). Contrary to that, we only allow agents to observe an \emph{existing} environmental signal. \emph{We do not modify the intrinsic or extrinsic rewards, design new features, or require a communication network}. Finally, boundary rules emerged in \cite{perolat2017multi} as well in the form of spatial territories. Such territories can increase inequality, while we maintain high levels of fairness.

% As a final note, showed that \cite{ota2020can} reinforcement learning problems with intrinsically low-dimensional state can benefit by intentionally increasing its dimensionality. To showcase our technique, we used the tragedy of the commons in CPR as a running example. Nevertheless, the technique is generic,  and could provide benefit to other coordination problems with low observability. \note{(DELETE?)}

\section{Agent \& Environment Models} \label{Models}

\subsection{Multi-Agent Reinforcement Learning} \label{Multi-Agent Reinforcement Learning}

We consider a decentralized multi-agent reinforcement learning scenario in a partially observable general-sum Markov game \cite{Shapley1095}. At each time-step, agents take actions based on a partial observation of the state space, and receive an individual reward. Each agent learns a policy independently. More formally, let $\mathcal{N} = \{1, \dots, N\}$ denote the set of agents, and $\mathcal{M}$ be an $N$-player, partially observable Markov game defined on a set of states $\mathcal{S}$. An observation function $\mathcal{O}^n: \mathcal{S} \rightarrow \mathds{R}^d$ specifies agent $n$'s $d$-dimensional view of the state space. Let $\mathcal{A}^n$ denote the set of actions for agent $n \in \mathcal{N}$, and $\bm{a} = \times_{\forall n \in \mathcal{N}} a^n$, where $a^n \in \mathcal{A}^n$, the joint action. The states change according to a transition function $\mathcal{T}: \mathcal{S} \times \mathcal{A}^1 \times \dots \times \mathcal{A}^N \rightarrow \Delta(\mathcal{S})$, where $\Delta(\mathcal{S})$ denotes the set of discrete probability distributions over $\mathcal{S}$. Every agent $n$ receives an individual reward based on the current state $\sigma_t \in \mathcal{S}$ and joint action $\bm{a}_t$. The latter is given by the reward function $r^n : \mathcal{S} \times \mathcal{A}^1 \times \dots \times \mathcal{A}^N \rightarrow \mathds{R}$. Finally, each agent learns a policy $\pi ^n : \mathcal{O}^n \rightarrow \Delta(\mathcal{A}^n)$ independently through their own experience of the environment (observations and rewards). Let $\bm{\pi} = \times_{\forall n \in \mathcal{N}} \pi^n$ denote the joint policy. The goal for each agent is to maximize the long term discounted payoff, as given by $V^n_{\bm{\pi}} (\sigma_0) = \mathds{E} \left[ \sum_{t = 0}^\infty \gamma^t r^n(\sigma_t, \bm{a}_t) | \bm{a}_t \sim \bm{\pi}_t, \sigma_{t+1} \sim \mathcal{T}(\sigma_t, \bm{a}_t) \right]$, where $\gamma$ is the discount factor and $\sigma_0$ is the initial state.

% \begin{equation} \label{Eq: payoff}
% 	V^n_{\bm{\pi}} (s_0) = \mathds{E} \left[ \underset{t = 0}{\overset{\infty}{\sum}} \gamma^t r^n(\sigma_t, \bm{a}_t) | \bm{a}_t \sim \bm{\pi}_t, s_{t+1} \sim \mathcal{T}(\sigma_t, \bm{a}_t) \right]
% \end{equation}

\subsection{The Common Fishery Model} \label{Fishery Model}

In order to better understand the impact of self-interested appropriation, it would be beneficial to examine the dynamics of \emph{real-world} common-pool renewable resources. To that end, we present an abstracted bio-economic model for commercial fisheries \cite{clark2006worldwide,diekert2012tragedy}. The model describes the dynamics of the stock of a common-pool renewable resource, as a group of appropriators harvest over time. The harvest depends on (i) the effort exerted by the agents and (ii) the ease of harvesting a resource at that point of time, which depends on the stock level. The stock replenishes over time with a rate dependent on the current stock level.

More formally, let $\mathcal{N}$ denote the set of appropriators, $\epsilon_{n, t} \in [0, \mathcal{E}_{max}]$ the effort exerted by agent $n$ at time-step $t$, and $E_t = \sum_{n \in \mathcal{N}} \epsilon_{n, t}$ the total effort at time-step $t$. The total harvest is given by Eq. \ref{Eq: Total harvest}, where $s_{t} \in [0,\infty)$ denotes the stock level (i.e., amount of resources) at time-step $t$, $q(\cdot)$ denotes the catchability coefficient (Eq. \ref{Eq: Catchability coefficient function}), and $S_{eq}$ is the equilibrium stock of the resource.
\begin{equation} \label{Eq: Total harvest}
	H(E_t,s_t) =
	\begin{cases}
		q(s_t) E_t &, \text{if } q(s_t) E_t \leq s_t \\
		s_t  &, \text{otherwise} \\
	\end{cases}
\end{equation}
\begin{equation} \label{Eq: Catchability coefficient function}
	q(x) = 
	\begin{cases}
		\frac{x}{2 S_{eq}} &, \text{if } x \leq 2 S_{eq} \\
		1  &, \text{otherwise} \\
	\end{cases}
\end{equation}

Each environment can only sustain a finite amount of stock. If left unharvested, the stock will stabilize at $S_{eq}$. Note also that $q(\cdot)$, and therefore $H(\cdot)$, are proportional to the current stock, i.e., the higher the stock, the larger the harvest for the same total effort. The stock dynamics are governed by Eq. \ref{Eq: stock dynamics}, where $F(\cdot)$ is the spawner-recruit function (Eq. \ref{Eq: Spawner-recruit function}) which governs the natural growth of the resource, and $r$ is the growth rate. To avoid highly skewed growth models and unstable environments (`behavioral sink' \cite{calhoun1962population,doi:10.1177/00359157730661P202}), $r \in [-W(-1/(2 e)), -W_{-1}(-1/(2 e))] \approx [0.232, 2.678]$, where $W_k(\cdot)$ is the Lambert $W$ function (see Section \ref{supplement: Growth Rate} for details).
\begin{equation} \label{Eq: stock dynamics}
	s_{t+1} = F(s_t - H(E_t,s_t))
\end{equation}
\begin{equation} \label{Eq: Spawner-recruit function}
	F(x) = x e^{r(1 - \frac{x}{S_{eq}})}
\end{equation}

We assume that the individual harvest is proportional to the exerted effort (Eq. \ref{Eq: Individual harvest}), and the revenue of each appropriator is given by Eq. \ref{Eq: Revenue}, where $p_t$ is the price (\$ per unit of resource), and $c_t$ is the cost (\$) of harvesting (e.g., operational cost, taxes, etc.). Here lies the `tragedy': the benefits from harvesting are private ($p_t h_{n, t}(\epsilon_{n, t},s_t)$), but the loss is borne by all (in terms of a reduced stock, see Eq. \ref{Eq: stock dynamics}).
\begin{equation} \label{Eq: Individual harvest}
	h_{n, t}(\epsilon_{n, t},s_t) = \frac{\epsilon_{n, t}}{E_t} H(E_t,s_t)
\end{equation}
\begin{equation} \label{Eq: Revenue}
	u_{n, t}(\epsilon_{n, t},s_t) = p_t h_{n, t}(\epsilon_{n, t},s_t) - c_t
\end{equation}

\subsubsection{\textbf{Optimal Harvesting}} \label{Optimal Harvesting}

The question that naturally arises is: what is the `optimal' effort in order to harvest a yield that maximizes the revenue (Eq. \ref{Eq: Revenue}). We make two assumptions: First, we assume that the entire resource is owned by a single entity (e.g., a firm or the government), which possesses complete knowledge of and control over the resource. Thus, we only have a single control variable, $E_t$. This does not change the underlying problem since the total harvested resources are linear in the proportion of efforts put by individual agents (Eq. \ref{Eq: Individual harvest}). Second, we consider the case of zero discounting, i.e., future revenues are weighted equally with current ones. Of course firms (and individuals) do discount the future and bio-economic models should take that into account, but this complicates the analysis and it is out of the scope of this work. We argue we can still draw useful insight into the problem.
% \footnote{`Optimal' is used in a technical sense, as the strategy that maximizes the revenue subject to the model equations, and it does not carry any moralistic implications.}

Our control problem consists of of finding a piecewise continuous control $E_t$, so as to maximize the total revenue ($\max_{E_t} \sum_{t = 0}^T U_{t}(E_{t},s_t)$). The maximization problem can be solved using Optimal Control Theory \cite{ding2010introduction_optimalcontrol,lenhart2007optimal}. We have proven the following theorem:

% \begin{equation} \label{Eq: Cumulative Revenue}
% 	\underset{E_t}{\max} \underset{t = 0}{\overset{T}{\sum}} U_{t}(E_{t},s_t)
% \end{equation}

\begin{theorem} \label{Th: Optimal Harvesting Strategy}
	The optimal control variable $E^*_t$ that solves the maximization problem $\max_{E_t} \sum_{t = 0}^T U_{t}(E_{t},s_t)$ given the model dynamics described in Section \ref{Fishery Model} is given by Eq. \ref{Eq: theorem optimal lambda}, where $\lambda_t$ are the adjoint variables of the Hamiltonians:
	\begin{equation} \label{Eq: theorem optimal lambda}
		E^*_{t+1} =
		\begin{cases}
			E_{max} ,& \text{if } (p_{t+1}-\lambda_{t+1})q(F(s_t - H(E_t, s_t))) \geq 0 \\
			0  ,& \text{if } (p_{t+1}-\lambda_{t+1})q(F(s_t - H(E_t, s_t))) < 0 \\
		\end{cases}
	\end{equation}
\end{theorem}

\begin{proof}
	(sketch) We formulate the Hamiltonians \cite{lenhart2007optimal,ding2010introduction_optimalcontrol}, which turn out to be linear in the control variables $E_{t+1}$ with coefficients $(p_{t+1}-\lambda_{t+1})q(F(s_t - H(E_t, s_t)))$. Thus, the optimal sequence of $E_{t+1}$ that maximizes the Hamiltonians is given according to the sign of those coefficients. See Section \ref{supplement: Proof of Theorem 2.1} for the complete proof.
\end{proof}

The optimal strategy is a bang–bang controller, which switches based on the adjoint variable values, stock level, and price. The values for $\lambda_t$ do not have a closed form expression (because of the discontinuity of the control), but can be found iteratively for a given set of environment parameters ($r$, $S_{eq}$) and the adjoint equations \cite{lenhart2007optimal,ding2010introduction_optimalcontrol}. However, the discontinuity in the control input makes solving the adjoint equations quite cumbersome. We can utilize iterative forward/backward methods as in \cite{ding2010introduction_optimalcontrol}, but this is out of the scope of this paper.

There are a few interesting key points. First, to compute the optimal level of effort we require observability of the resource stock, which is not always a realistic assumption (in fact in this work we do not make this assumption). Second, we require complete knowledge of the strategies of the other appropriators. Third, even if both the aforementioned conditions are met, the bang-bang controller of Eq. \ref{Eq: theorem optimal lambda} does not have a constant transition limit; the limit changes at each time time-step, determined by the adjoint variable $\lambda_{t+1}$, thus finding the switch times remains quite challenging.

\subsubsection{\textbf{Harvesting at Maximum Effort}} \label{Harvesting at Maximum Effort}

To gain a deeper understanding of the dynamics of the environment, we will now consider a baseline strategy where every agent harvests with the maximum effort at every time-step, i.e., $\epsilon_{n, t} = \mathcal{E}_{max}, \forall n \in \mathcal{N}, \forall t$. This corresponds to the Nash Equilibrium of a stage game (myopic agents). For a constant growth rate $r$ and a given number of agents $N$, we can identify two interesting stock equilibrium points ($S_{eq}$): the `limit of sustainable harvesting', and the `limit of immediate depletion'.

The limit of sustainable harvesting ($S^{\scriptscriptstyle N, r}_{\scriptscriptstyle LSH}$) is the stock equilibrium point where the goal of sustainable harvesting becomes trivial: for any $S_{eq} > S^{\scriptscriptstyle N, r}_{\scriptscriptstyle LSH}$, the resource will not get depleted, even if all agents harvest at maximum effort. Note that the coordination problem \emph{remains far from trivial} even for $S_{eq} > S^{\scriptscriptstyle N, r}_{\scriptscriptstyle LSH}$, especially for increasing population sizes. Exerting maximum effort in environments with $S_{eq}$ close to $S^{\scriptscriptstyle N, r}_{\scriptscriptstyle LSH}$ will yield low returns because the stock remains low, resulting in a small catchability coefficient. In fact, this can be seen in Fig. \ref{fig:max_effort_rewards_all} which depicts the social welfare (SW), i.e., sum of utilities, against increasing $S_{eq}$ values ($N \in [2, 64]$, $\mathcal{E}_{max} = 1$, $r = 1$). Red dots\footnote{Slight deviations from the predicted theoretical values of Eq. \ref{Eq: LSH} due to the finite episode length and non-zero threshold.} denote the $S^{\scriptscriptstyle N, r}_{\scriptscriptstyle LSH}$. Thus, \emph{the challenge is not only to keep the strategy sustainable, but to keep the resource stock high, so that the returns can be high as well}.

% trim={<left> <lower> <right> <upper>}
\begin{figure}[t!]
	\centering
	\includegraphics[width = 1\linewidth, trim={0em 0em 0em 0em}, clip]{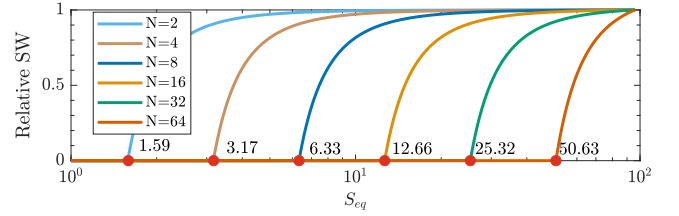}
	\caption{Social welfare (SW) -- normalized by the maximum SW obtained in each setting -- against increasing $S_{eq}$ values. $N \in [2, 64]$, $\mathcal{E}_{max} = 1$, and $r = 1$. $x$-axis is in logarithmic scale.}
	\Description{Social welfare against increasing stock equilibrium values.}
	\label{fig:max_effort_rewards_all}
\end{figure}

On the other end of the spectrum, the limit of immediate depletion ($S^{\scriptscriptstyle N, r}_{\scriptscriptstyle LID}$) is the stock equilibrium point where the resource is depleted in one time-step (under maximum harvest effort by all the agents). The problem does not become impossible for $S_{eq} \leq S^{\scriptscriptstyle N, r}_{\scriptscriptstyle LID}$, yet, \emph{exploration can have catastrophic effects} (amplifying the problem of global exploration in MARL). The following two theorems prove the formulas for $S^{\scriptscriptstyle N, r}_{\scriptscriptstyle LSH}$ and $S^{\scriptscriptstyle N, r}_{\scriptscriptstyle LID}$.

\begin{theorem}
	The limit of sustainable harvesting $S^{\scriptscriptstyle N, r}_{\scriptscriptstyle LSH}$ for a continuous resource governed by the dynamics of Section \ref{Fishery Model}, assuming that all appropriators harvest with the maximum effort $\mathcal{E}_{max}$, is:
	\begin{equation} \label{Eq: LSH}
		S^{\scriptscriptstyle N, r}_{\scriptscriptstyle LSH} = \frac{e^r N \mathcal{E}_{max}}{2 (e^r - 1)}
	\end{equation}
	% \begin{equation} \label{Eq: LSH}
	% 	S^{\scriptscriptstyle N, r}_{\scriptscriptstyle LSH} = \frac{e^r N \mathcal{E}_{max}}{2 (e^r - 1)} + \underset{x \rightarrow 0}{\lim} x
	% \end{equation}
\end{theorem}

\begin{proof}
Note that for $S_{eq} > S^{\scriptscriptstyle N, r}_{\scriptscriptstyle LSH}$, $q(s_t)E_t < s_t, \forall t$, otherwise the resource would be depleted. Moreover, if $s_0 = S_{eq}$ -- which is a natural assumption, since prior to any intervention the stock will have stabilized on the fixed point -- then\footnote{Given that $r \in [-W(-1/(2 e)), -W_{-1}(-1/(2 e))]$.} $s_t < 2S_{eq}, \forall t$. Thus, we can re-write Eq. \ref{Eq: Total harvest} and \ref{Eq: Catchability coefficient function} as:
\begin{equation*} \label{Eq: lsh proof 0}
	H(E_{t},s_t) = \frac{s_t}{2 S_{eq}} E_t = \frac{s_t N \mathcal{E}_{max}}{2 S_{eq}}
\end{equation*}

Let $\alpha \triangleq \frac{N \mathcal{E}_{max}}{2 S_{eq}}$, and $\beta = 1 - \alpha$. The state transition becomes:
\begin{equation*} \label{Eq: lsh proof 1}
	s_{t+1} = F(s_t - \alpha s_t) = \beta s_t e^{r(1-\frac{\beta}{S_{eq}}s_t)}
\end{equation*}

We write it as a difference equation:
\begin{equation*} \label{Eq: lsh derive 2}
	\Delta_t(s_t) \triangleq s_{t+1}-s_t = (\beta e^{r(1-\frac{\beta}{S_{eq}}s_t)}-1)s_t
\end{equation*}

At the limit of sustainable harvesting, as the stock diminishes to\footnote{In practice, $\delta$ is enforced by the granularity of the resource.} $s_t = \delta \rightarrow 0$, to remain sustainable it must be that $\Delta_t(s_t) > 0$. Thus, it must be that:
\begin{equation*} \label{Eq: lsh proof 3}
	\underset{s_t \rightarrow 0^+}{\lim} sgn(\Delta_t(s_t)) > 0 \overset{s_t \rightarrow 0^+ > 0}{\Rightarrow} \beta e^r - 1 > 0 \Rightarrow S_{eq} > \frac{e^r N \mathcal{E}_{max}}{2 (e^r - 1)}
\end{equation*}
\end{proof}

\begin{theorem}
	The limit of immediate depletion $S^{\scriptscriptstyle N, r}_{\scriptscriptstyle LID}$ for a continuous resource governed by the dynamics of Section \ref{Fishery Model}, assuming that all appropriators harvest with the maximum effort $\mathcal{E}_{max}$, is given by:
	% \begin{equation} \label{Eq: LID}
	% 	S^{\scriptscriptstyle N, r}_{\scriptscriptstyle LID} \leq \frac{N \mathcal{E}_{max}}{2}
	% \end{equation}
	\begin{equation} \label{Eq: LID}
		S^{\scriptscriptstyle N, r}_{\scriptscriptstyle LID} = \frac{N \mathcal{E}_{max}}{2}
	\end{equation}
\end{theorem}

\begin{proof}
The resource is depleted if:
\begin{equation*} \label{Eq: lid proof 0}
	H(E_{t},s_t) = s_t \Rightarrow q(s_t) E_t \geq s_t \Rightarrow \frac{s_t}{2 S_{eq}} E_t \geq s_t \Rightarrow S_{eq} \leq \frac{N \mathcal{E}_{max}}{2}
\end{equation*}
\end{proof}

\subsection{Environmental Signal} \label{Environmental Signal}

We introduce an auxiliary signal; side information from the environment (e.g., time, date etc.) that agents can potentially use in order to facilitate coordination and reach more sustainable strategies. Real-world examples include shepherds that graze on particular days of the week or fishermen that fish on particular months. In our case, the signal can be thought as a mechanism to increase the set of possible (individual and joint) policies. Such signals are \emph{amply available to the agents} \cite{hart2000simple,Danassis:2019:CMC:3306127.3331754}. \emph{We do not assume any a priori relation between the signal and the problem at hand}. In fact, in this paper we use a set of arbitrary integers, that repeat periodically. We use $\mathcal{G} = \{1, \dots, G\}$ to denote the set of signal values.

\section{Simulation Results} \label{Simulation Results}

\subsection{Setup} \label{Setup}

\subsubsection{\textbf{Environment Settings}} \label{Environment Settings}

Let $p_t = 1$, and $c_t = 0$, $\forall t$. We set the growth rate at $r = 1$, the initial population at $s_0 = S_{eq}$, and the maximum effort at $\mathcal{E}_{max} = 1$. The findings of Section \ref{Harvesting at Maximum Effort} provide a guide on the selection of the $S_{eq}$ values. Specifically we simulated environments with $S_{eq}$ given by Eq. \ref{Eq: seq choices}, where $K = \frac{S^{\scriptscriptstyle N, r}_{\scriptscriptstyle LSH}}{N} = \frac{e^r \mathcal{E}_{max}}{2 (e^r - 1)} \approx 0.79$ is a constant and $M_s \in \mathbb{R}^+$ is a multiplier that adjusts the scarcity (difficulty). $M_s = 1$ corresponds to $S_{eq} = S^{\scriptscriptstyle N, r}_{\scriptscriptstyle LSH}$.
\begin{equation} \label{Eq: seq choices}
	S_{eq} = M_s K N
\end{equation}

\subsubsection{\textbf{Agent Architecture}} \label{Agent Architecture}

Each agent uses a two-layer ($64$ neurons each) neural network for the policy approximation. The input (observation $o^n = \mathcal{O}^n(S)$) is a tuple $\langle \epsilon_{n, t-1}, u_{n, t-1}(\epsilon_{n, t-1}, s_{t-1}), g_t \rangle$ consisting of the individual effort exerted and reward obtained in the previous time-step and the current signal value. The output is a continuous action value $a_t = \epsilon_{n, t} \in [0, \mathcal{E}_{max}]$ specifying the current effort level. The policies are trained using the Proximal Policy Optimization (PPO) algorithm \cite{DBLP:journals/corr/SchulmanWDRK17}. PPO was chosen because it avoids large policy updates, ensuring a smoother training, and avoiding catastrophic failures. The reward received from the environment corresponds to the revenue, i.e., $r^n(\sigma_t, \bm{a}_t) = u_{n, t}(\epsilon_{n, t},s_t)$, and the discount factor was set to $\gamma = 0.99$.

% Laboratory experiments of appropriation dilemmas have showed that humans are fallible, have a bounded rationality, and norm-using. In complex, dynamic settings, they are not able to perform complete analysis before taking an action, rather individuals use heuristics, learn from their mistakes, and develop rules to improve the structure of the repetitive situations they face \cite{ostrom1999coping}. In addition, this class of reinforcement learning is considered a candidate theory of animal habit learning \cite{leibo2017multi} \cite{niv2009reinforcement}.

\subsubsection{\textbf{Signal Implementation}} \label{Signal Implementation}

The signal is represented as a $G$-dimensional one-hot encoded vector, where the high bit is shifted periodically. The initial value was chosen at random at the beginning of each episode to avoid bias towards particular values. Throughout this paper, the term \textit{no signal} will be used interchangeably to a unit signal size $G=1$, since a signal of size $1$ in one-hot encoding is just a constant input that yields no information. We evaluated signals of varying cardinality (see Section \ref{Results Robustness to Signal Size}).

\subsubsection{\textbf{Termination Condition}} \label{Termination Condition}

An episode terminates when either (a) the resource stock falls below a threshold $\delta = 10^{-4}$, or (b) a fixed number of time-steps $T_{max} = 500$ is reached. We trained our agents for a maximum of $5000$ episodes, with the possibility of early stopping if both of the following conditions are satisfied: (i) a minimum of $95\%$ of the maximum episode duration (i.e., $475$ time-steps) is reached for $200$ episodes in a row, and, (ii) the average total reward obtained by agents in each episode of the aforementioned $200$ episodes does not change by more than $5\%$. In case of early stopping, the metric values for the remainder of the episodes are extrapolated as the average of the last $200$ episodes, in order to properly average across trials.

% \subsubsection{\textbf{Evaluation Metrics}} \label{Evaluation Metrics}

\subsubsection{\textbf{Measuring The Influence of the Signal}} \label{Measuring The Influence of the Signal}

It is important to have a quantitative measure of the influence of the introduced signal. As such, we adapted the Causal Influence of Communication (CIC) \cite{lowe2019pitfalls} metric, initially designed to measure positive listening in emergent inter-agent communication. The CIC is calculated using the mutual information between the signal and the agent's action. Please see Section \ref{supplement: Causal Influence of Communication} for a complete description.

\subsubsection{\textbf{Reproducibility, Reporting of Results, Limitations}} \label{Reproducibility, Reporting of Results, Limitations}

% 8 trial, each trial is one training run with 5000 episodes, each episode has 500 time-steps.

Reproducibility is a major challenge in (MA)DRL due to different sources of stochasticity, e.g., hyper-parameters, model architecture, implementation details, etc. \cite{AAAI1816669,hernandez2019survey,Engstrom2020Implementation}. To minimize those sources, the implementation was done using RLlib\footnote{\url{https://docs.ray.io/en/latest/rllib.html}}, an open-source library for MADRL \cite{Liang2017Ray}. We refer the reader to Section \ref{supplement: Modeling Details} for a description of the architecture and hyper-parameters.

All simulations were \emph{repeated $8$ times} and the reported results are the average values of the last 10 episodes over those trials (excluding Fig. \ref{fig:strategies} which depicts a representative trial). (MA)DRL also lacks common practices for statistical testing \cite{AAAI1816669,hernandez2019survey}. In this work, we opted to use the Student's T-test \cite{student1908probable} due to it's robustness \cite{colas2019hitchhiker}. Nearly all of the reported results have p-values $< 0.05$.

Finally, we strongly believe that the community would benefit from reporting negative results. As such, we want to make clear that the proposed solution is not a panacea for all multi-agent coordination problems, not even for the proposed domain. For example, we failed to find sustainable policies using DDPG \cite{lillicrap2015continuous} -- with or without the signal -- for any set of environment parameters. This also comes to show the difficulty of the problem at hand. We suspect that the clipping in PPO's policy changes plays an important role in averting catastrophic failures in high-stakes environments.

\subsection{Results} \label{Results}

% NOTE ABOUT METHODS SECTION: The methods section should include metric definitions, results are all averages of 8 runs unless stated otherwise, experiments done from 2 to 64 agents, M ranges tested (0.2 to 1.2 and 0.2 to 2.0 for 8-16 and 32-64). Criterions for plotting change percentages.

% NOTE ABOUT M VALUES: LSH multiplier is 0.785 and in the environment we simulated, M=1 equals to Seq = LSH itself.

We present the result from a systematic evaluation of the proposed approach on a wide variety of environmental settings ($M_s \in [0.2, 1.2]$, i.e., ranging from way below the limit of immediate depletion, $M_s^{\scriptscriptstyle LID} \approx 0.63$, to above the limit of sustainable harvesting, $M_s^{\scriptscriptstyle LSH} = 1$) and population size ($N \in [2, 64]$). Due to lack of space we only present the most relevant results; see Section \ref{supplement: Simulation Results} for a complete report (e.g., results tables, fairness, small population sizes, etc.). 

In the majority of the results, we study the influence of a signal of cardinality $G=N$ compared to no signal ($G=1$). Thus, unless stated otherwise, the term `with signal' will refer to $G=N$.

\subsection{Sustainability \& Social Welfare} \label{Results Sustainability & Social Welfare}

\subsubsection{\textbf{Sustainability}} \label{Results Sustainability}

% Begin: Episode Lengths
% trim={<left> <lower> <right> <upper>}
\begin{figure}[t!]
	\centering
	\begin{subfigure}[t]{1\linewidth}
		\centering
		\begin{minipage}{0.85\linewidth}
		\includegraphics[width = 1\linewidth, trim={0em 2em 0em 0em}, clip]{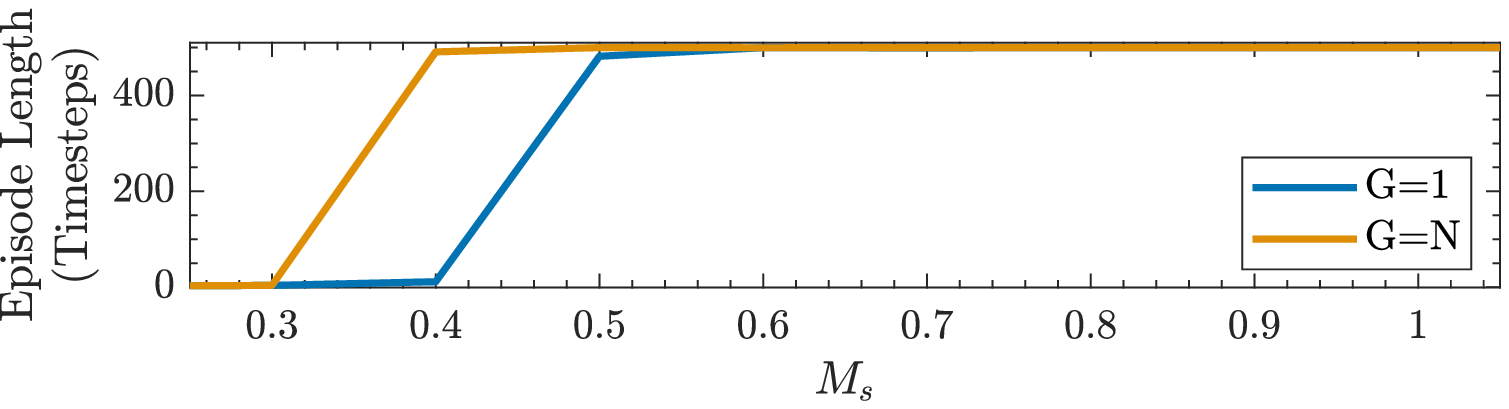}
		\end{minipage}\hfill
		\begin{minipage}{0.15\linewidth}
		\caption{$N=8$}
		\label{fig:eplens_n8}
		\end{minipage}
	\end{subfigure}
	\hfill
	\begin{subfigure}[t]{1\linewidth}
		\centering
		\begin{minipage}{0.85\linewidth}
		\includegraphics[width = 1\linewidth, trim={0em 2em 0em 0em}, clip]{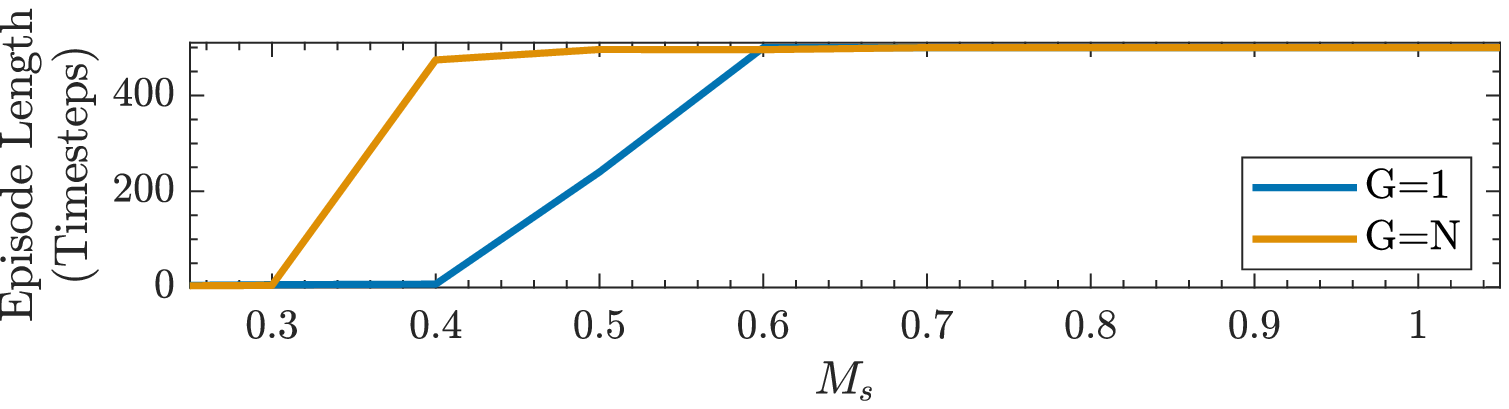}
		\end{minipage}\hfill
		\begin{minipage}{0.15\linewidth}
		\caption{$N=16$}
		\label{fig:eplens_n16}
		\end{minipage}
	\end{subfigure}
	\hfill
	\begin{subfigure}[t]{1\linewidth}
		\centering
		\begin{minipage}{0.85\linewidth}
		\includegraphics[width = 1\linewidth, trim={0em 2em 0em 0em}, clip]{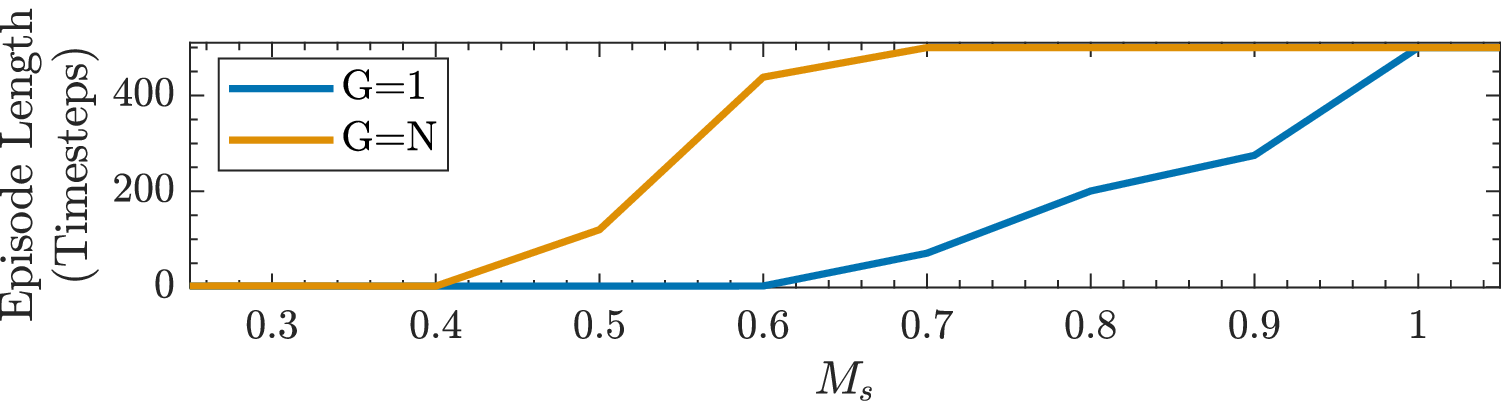}
		\end{minipage}\hfill
		\begin{minipage}{0.15\linewidth}
		\caption{$N=32$}
		\label{fig:eplens_n32}
		\end{minipage}
	\end{subfigure}
	\hfill
	\begin{subfigure}[t]{1\linewidth}
		\centering
		\begin{minipage}{0.85\linewidth}
		\includegraphics[width = 1\linewidth, trim={0em 0em 0em 0em}, clip]{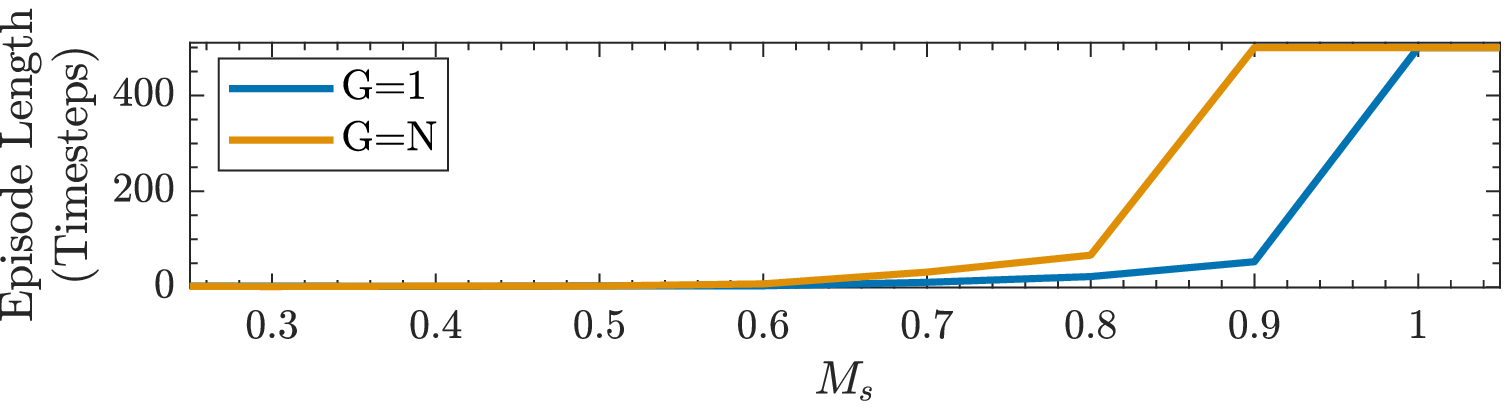}
		\end{minipage}\hfill
		\begin{minipage}{0.15\linewidth}
		\caption{$N=64$}
		\label{fig:eplens_n64}
		\end{minipage}
	\end{subfigure}
	\hfill
	
	\caption{Episode length, with and without the signal ($G=N$), for environments of decreasing difficulty (increasing equilibrium stock multiplier $M_s$).}
	\Description{Episode length, with and without the signal, for environments of decreasing difficulty}
	\label{fig:episode_lengths}
\end{figure}
% End: Episode Lengths

We declare a strategy `sustainable', iff the agents reach the maximum episode duration ($500$ steps), i.e., they do not deplete the resource. Fig. \ref{fig:episode_lengths} depicts the achieved episode length -- with and without the presence of a signal ($G=N$) -- for environments of decreasing difficulty (increasing $S_{eq} \propto M_s$). The introduction of the signal significantly increases the range of environments ($M_s$) where sustainability can be achieved. Assuming that $M_s \in [0, 1]$ -- since for $M_s \geq 1$ sustainability is guaranteed by definition -- we have an increase of $17\% - 300\%$ ($46\%$ on average) in the range of sustainable $M_s$ values. Moreover, as the number of agents increases ($N = 32$ \& $64$), depletion is avoided in non-trivial $M_s$ values \emph{only with the introduction of the signal}. Finally, note that the $M_s$ value where a sustainable strategy is found increases with $N$, which demonstrates that the difficulty of the problem increases superlinearly to $N$ (given that $S_{eq} \propto M_s N$).

% N	Signal										No signal										Increase
% 8		Ms 0.4-1, i.e., in 7 out of 10 Ms values	Ms 0.5-1, i.e., in 6 out of 10 Ms values	17 \%
% 16	Ms 0.5-1, i.e., in 6 out of 10 Ms values	Ms 0.6-1, i.e., in 5 out of 10 Ms values	20 \%
% 32	Ms 0.7-1, i.e., in 4 out of 10 Ms values	Ms 1-1,   i.e., in 1 out of 10 Ms values	300 \%
% 64	Ms 0.9-1, i.e., in 2 out of 10 Ms values	Ms 1-1,   i.e., in 1 out of 10 Ms values	100 \%
% 
% In total:				  19 out of 40 Ms values					  13 out of 40 Ms values	46 \%

\subsubsection{\textbf{Social welfare}} \label{Results Social Welfare}

Reaching a sustainable strategy -- i.e., avoiding resource depletion -- is only one piece of the puzzle; an agent's revenue depends on the harvest (Eq. \ref{Eq: Total harvest}), which in turns depends on the catchability coefficient (Eq. \ref{Eq: Catchability coefficient function}). Thus, in order to achieve a high social welfare (sum of utilities, i.e., $\sum_{n \in \mathcal{N}} r^n(\cdot)$), the agents need to learn policies that balance the trade-off between maintaining a high stock (which ensues a high catchability coefficient), and yielding a large harvest (which results to a higher reward). This problem becomes even more apparent as resources become more abundant (i.e., for $M_s = 1 \pm x$, i.e., close to the limit of sustainable harvesting (below or, especially, \emph{above}), see Section \ref{Harvesting at Maximum Effort}). In these settings, it is easy to find a sustainable strategy; a myopic best-response strategy (harvesting at maximum effort) by all agents will not deplete the resource. Yet, it will result in low social welfare (SW).

Fig. \ref{fig:social_welfare_change} depicts the relative difference in SW, in a setting with and without the signal ($(SW_{G=N} - SW_{G=1}) / SW_{G=1}$, where $SW_{G=X}$ denotes the SW achieved using a signal of cardinality $X$), for environments of decreasing difficulty (increasing $S_{eq} \propto M_s$) and varying population size ($N \in [8, 64]$). To improve readability, changes greater than 100\% are shown with numbers on the top of the bars. Given the various sources of stochasticity, we opted to omit settings in which agents were not able to reach an episode duration of more than 10 time-steps (either with or without the signal).

The presence of the signal results in a significant improvement in SW. Specifically, we have an average of $258\%$ improvement \emph{across all the depicted settings}\footnote{The averaging is performed across the entire range of the depicted $M_s \in [0.2, 1.2]$, including the really scarce environments of $M_s = 0.2$ and $0.3$ where there is no sustainable strategy with or without the signal and, thus, the change is zero.} in Fig. \ref{fig:social_welfare_change}, while the maximum improvement is $3306\%$. These improvements stem from (i) achieving more sustainable strategies, and (ii) improved cooperation. The former results in higher rewards due to longer episodes in settings where the strategies without the signal deplete the resource. The latter allows to avoid over-harvesting, which results in higher catchability coefficient, in settings where both strategies (with, or without the signal) are sustainable. The contribution of the signal is much more pronounced under scarcity: the difference in achieved SW decreases as $M_s$ increases, eventually becoming less than $10\%$ ($M_s>1$ for $N=8$ \& $16$, and $M_s>1.2$ for $N=32$ \& $64$). This suggests that the proposed approach is of high practical value in environments where resources are \emph{scarce} (like most real-world applications), a claim that we further corroborate in Sections \ref{Results CIC} and \ref{Results Emergence of Temporal Conventions}.

% Begin: Social welfare changes
% trim={<left> <lower> <right> <upper>}
\begin{figure}[t!]
	\centering
	\includegraphics[width = 1\linewidth, trim={0em 0em 0em 0em}, clip]{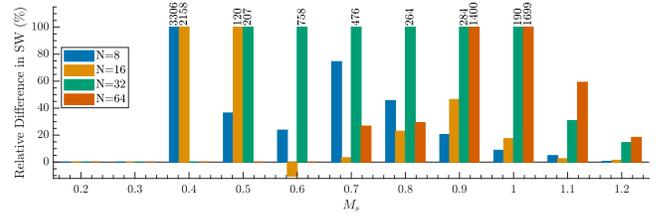}
	\caption{Relative difference in social welfare (SW) when signal of cardinality $G=N$ is introduced ($(SW_{G=N} - SW_{G=1}) / SW_{G=1}$, where $SW_{G=X}$ denotes the SW achieved using a signal of cardinality $X$), for environments of decreasing difficulty (increasing $S_{eq} \propto M_s$) and varying population size ($N \in [4, 64]$). To improve readability, changes greater than 100\% are shown with numbers on the top of the bars.}
	\Description{Relative difference in social welfare when signal of cardinality N is introduced.}
	\label{fig:social_welfare_change}
\end{figure}
% End: Social welfare changes

% NOTE:
% Q: The case with N=16 and M_s=0.6 is quite confusing. Can you explain why the use of signal turned out to be actually harmful here? (Similar thing happens with N=16 also in the cases where M_s is 0.7 and 0.8.)
% A: At this point both methods manage to find sustainable strategies, and the small difference is only due to noise, nothing else (this is verified by the p-values).

% Begin: Convergence speed changes
% trim={<left> <lower> <right> <upper>}
\begin{figure}[t!]
	\centering
	\includegraphics[width = 1\linewidth, trim={0em 0em 0em 0em}, clip]{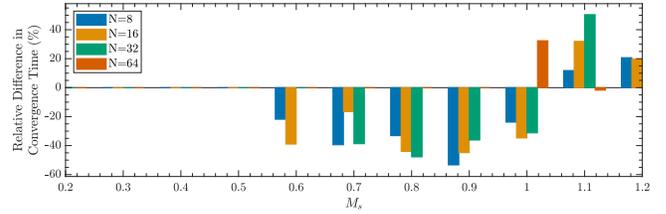}
	\caption{Relative difference in convergence time with the introduction of a signal ($(CT_{G=N} - CT_{G=1}) / CT_{G=1}$, where $CT_{G=X}$ denotes the time until convergence when using a signal of cardinality $X$), for environments of decreasing difficulty (increasing $S_{eq} \propto M_s$) and varying population size ($N \in [8, 64]$)}
	\Description{Relative difference in convergence time with the introduction of a signal.}
	\label{fig:convergence_time_change}
\end{figure}
% End: Convergence speed changes

\subsection{Convergence Speed} \label{Results Convergence Speed}

The second major influence of the introduction of the proposed signal -- besides the sustainability and efficiency of the learned strategies -- is on the convergence time. Let the system be considered converged when the global state does not change significantly. As a practical way to pinpoint the time of convergence, we used the `Termination Criterion' of Section \ref{Termination Condition}. Fig. \ref{fig:convergence_time_change} depicts the relative difference in convergence time with the introduction of a signal ($(CT_{G=N} - CT_{G=1}) / CT_{G=1}$, where $CT_{G=X}$ denotes the time until convergence, in \#episodes, when using a signal of cardinality $X$), for environments of decreasing difficulty (increasing $S_{eq} \propto M_s$) and varying population size ($N \in [8, 64]$). We have omitted the settings in which agents were not able to reach an episode duration of more than 10 time-steps (either with or without the signal).

There is a disjoint effect of the signal on the convergence speed. Up to the limit of sustainable harvesting ($M_s \leq 1$), the signal significantly improves the convergence speed ($13\%$ improvement on average, \emph{across all the depicted settings} including the ones with no improvement, and up to $53\%$). This is vital, as the majority of \emph{real-world problems involve managing scarce resources}. On the other hand, for $M_s > 1$, i.e., settings with abundant resources, the system converges faster without the signal ($14\%$ slower with the signal on average, across all the depicted settings). One possible explanation is that as resources become more abundant, it is harder (impossible for $M_s > 1$) for agents to deplete them. Therefore the learning is more efficient -- and the convergence is faster -- since the episodes tend to last longer (without needing the signal). Moreover, having an abundance of resources decouples the effects of the agents' actions to each other, reducing the variance, and again making easing the learning process without the signal.

\subsection{Influence of Signal on Agent Strategies} \label{Results CIC}

% Begin: CIC
% trim={<left> <lower> <right> <upper>}
\begin{figure}[t!]
	\centering
	\includegraphics[width = 1\linewidth, trim={0em 0em 0em 0em}, clip]{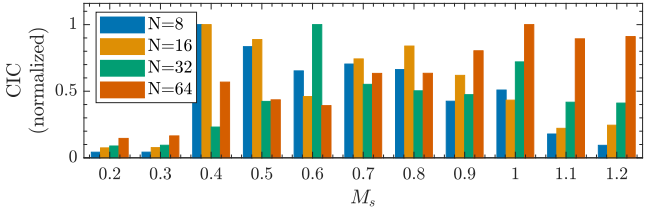}
	\caption{Average (over agents and trials) CIC values (normalized) vs. the equilibrium stock multiplier $M_s$, for population/signal size $N=G \in \{8, 16, 32 ,64\}$.}
	\Description{CIC values, averaged over agents and trials, vs. the equilibrium stock multiplier.}
	\label{fig:cic_all_normalised}
\end{figure}
% End: CIC

The results presented so far provide a qualitative measure of the influence of the introduced signal through the improvement on sustainability, social welfare, and convergence speed. They also indicate a decrease on the influence of the signal as resources become abundant. The question that naturally arises is: how much do agents actually take the signal into account in their policies? To answer this question, Fig. \ref{fig:cic_all_normalised} depicts the CIC values -- a quantitative measure of the influence of the introduced signal (see Section \ref{Measuring The Influence of the Signal}) -- versus increasing values of $M_s$ (i.e, increasing $S_{eq} \propto M_s$, or more abundant resources), for population/signal size $N=G \in \{8, 16, 32 ,64\}$. The values are averaged across the 8 trials and the agents, and are normalized with respect to the maximum value for each population\footnote{For the absolute values please refer to Table \ref{tab: CIC values}. Fig. \ref{fig:cic_all_normalised} shows trends across $M_s$ values -- \emph{not} between populations sizes (due to the normalization).}. Higher CIC values indicate a higher causal influence of the signal.

%  (agents in different population sizes have different state spaces, environments, can learn more complex policies due to the larger number of signals, etc. See the CIC implementation in \cite{supplement})

CIC is low for the trials in which a sustainable strategy could not be found ($M_s=0.2 - 0.3$ for $N = 8$, $16$, $M_s=0.2 - 0.5$ for $N = 32$, and $M_s=0.2 - 0.8$ for $N = 64$, see Fig. \ref{fig:episode_lengths}). In cases where a sustainable strategy was reached (e.g., $M_s \geq 0.4$ for $N = 8$), we see significantly higher CIC values on scarce resource environments, and then the CIC decreases as $M_s$ increases. \emph{The harder the coordination problem, the more the agents rely on the environmental signal.}

\subsection{Robustness to Signal Size} \label{Results Robustness to Signal Size}

% Begin: Large signal range results
% trim={<left> <lower> <right> <upper>}
\begin{figure}[t!]
	\centering
	\begin{subfigure}[t]{1\linewidth}
		\centering
		\begin{minipage}{0.9\linewidth}
		\includegraphics[width =1\linewidth, trim={0em 1.8em 0em 0em}, clip]{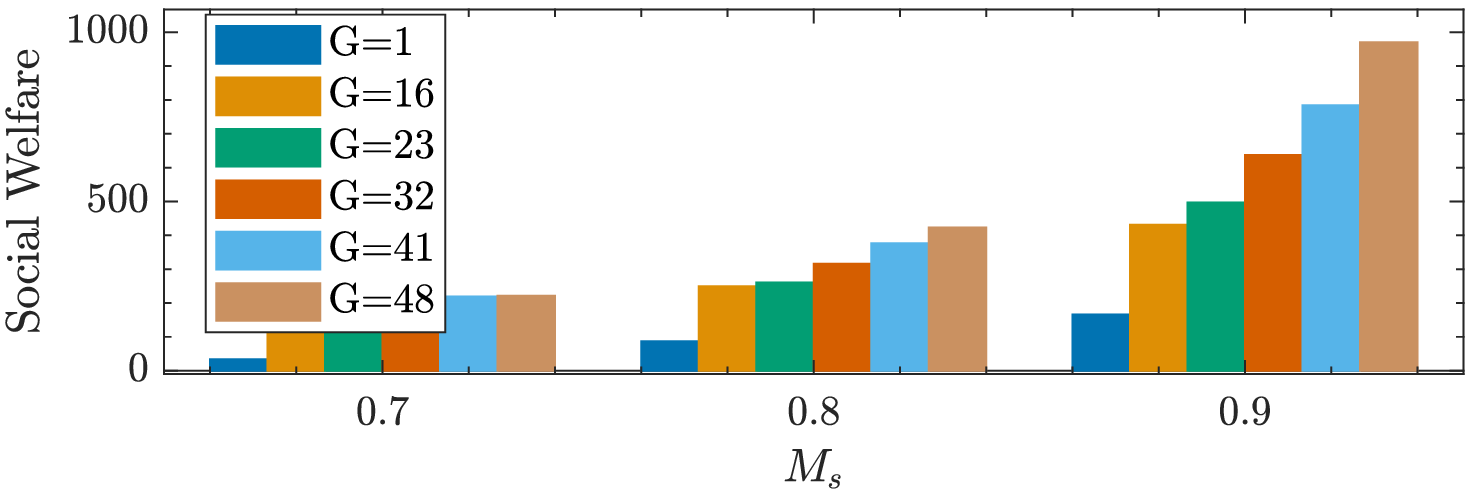}
		\end{minipage}\hfill
		\begin{minipage}{0.1\linewidth}
		\caption{ }
		\label{fig:rewards_lsr}
		\end{minipage}
	\end{subfigure}
	\hfill
	\begin{subfigure}[t]{1\linewidth}
		\centering
		\begin{minipage}{0.9\linewidth}
		\includegraphics[width = 1\linewidth, trim={0em 0em 0em 0em}, clip]{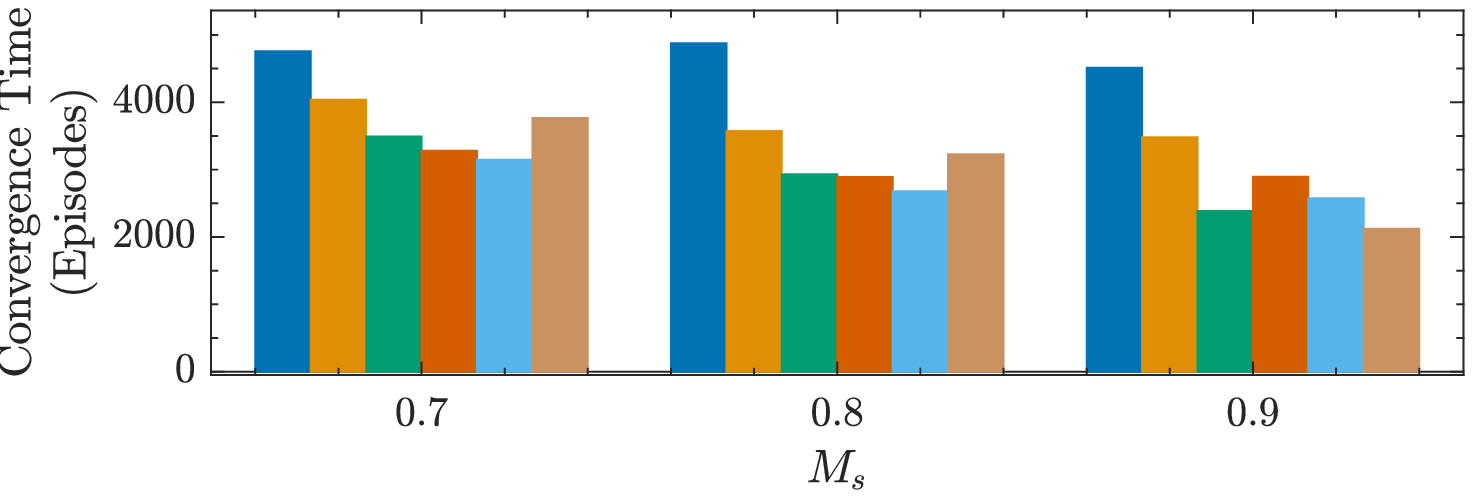}
		\end{minipage}\hfill
		\begin{minipage}{0.1\linewidth}
		\caption{ }
		\label{fig:numeps_lsr}
		\end{minipage}
	\end{subfigure}
	\hfill
	
	\caption{Achieved social welfare (Fig. \ref{fig:rewards_lsr}) and convergence time (Fig. \ref{fig:numeps_lsr}) for different signals of cardinality ($G$) $1$, $\frac{N}{2}=16$, $23$, $N=32$, $41$, and $\frac{3N}{2}=48$ (for varying resource levels $M_s$).}
	\Description{Achieved social welfare and convergence time for signal of different cardinality.}
	\label{fig:lsr_results}
\end{figure}
% End: Large signal range results

Up until now we have evaluated the presence (or lack thereof) of an environmental signal of cardinality equal to the population size ($G = N$). This requires exact knowledge of $N$, thus it is interesting to test the robustness of the proposed approach under varying signal sizes. As a representative test-case, we evaluated different signals of cardinality $G=1$, $\frac{N}{2}$, $23$, $N$, $41$, and $\frac{3N}{2}$ for $N=32$ and moderate scarcity for the resource ($M_s$ values of $0.7$, $0.8$ and $0.9$). The values $23$ and $41$ were chosen as they are prime numbers (i.e., \emph{not multiples} of $N$). Fig. \ref{fig:lsr_results} depicts the achieved social welfare and convergence time under the aforementioned settings.

Starting with Fig. \ref{fig:rewards_lsr} we can see that the SW increases with the signal cardinality. Specifically, we have $263\%$, $255\%$, $341\%$, $416\%$, and $474\%$ improvement on average across the three $M_s$ values for $G=\frac{N}{2}$, $23$, $N$, $41$, and $\frac{3N}{2}$, respectively. We hypothesize that the improvement stems from an increased joint strategy space that the larger signal size allows. A signal size larger than $N$ can also allow the emergence of `rest' (fallow) periods -- signal values where the majority of agents harvests at really low efforts. This would allow the resource to recuperate, and increase the SW through a higher catchability coefficient. See Section \ref{Results Emergence of Temporal Conventions} / Fig. \ref{fig:strategies_n2} for an example.
% an example of such a joint strategy.

Regarding the convergence speed (Fig. \ref{fig:numeps_lsr}), we have $22\%$, $38\%$, $36\%$, $41\%$, and $36\%$ improvement on average (across $M_s$ values).
 % for the aforementioned signal values, respectively.

% Contrary to the case of the SW, though, varying the size of the signal does not result in a consistent change.

These results showcase that the introduction of the signal itself -- regardless of its cardinality or, more generally, its temporal representative power -- provides a clear benefit to the agents in terms of SW and convergence speed. This greatly improves the real-world applicability of the proposed technique, as the the \emph{knowledge of the exact population size is not required}; instead the agents can opt to select any signal available in their environment\footnote{The signal is represented as a one-hot vector, i.e., Fig. \ref{fig:lsr_results} shows that a network with 32 inputs can work for population sizes $N \in [16, 48]$, or equivalently, that agents in a population of size $N=32$ can use networks with $16 - 48$ inputs for the signal.}. Moreover, the signal cardinality can also be considered as a design choice, depending on the requirements and limitations of the system.
% Moreover, the signal cardinality can also be considered as a design choice, decided depending on the requirements and limitations of the system in consideration.

\subsection{Emergence of Temporal Conventions} \label{Results Emergence of Temporal Conventions}

\subsubsection{\textbf{Qualitative Analysis}} \label{Qualitative Analysis}

We have seen that the introduction of an arbitrary signal facilitates cooperation and the sustainable harvesting. But do temporal conventions actually emerge?
% In this section we take a closer look at the learned strategies.

Fig. \ref{fig:strategies_n4} presents an example of the evolution of the agents' strategies for each signal value for a population of $N=4$, signal size $G=N=4$, and equilibrium stock multiplier $M_s=0.5$ (smoothed over 50 episodes). Each row represents an agent (agent $n_i$), while each column represents a signal value (value $g_j$). Each line represents the average effort the agent exerts on that specific signal value -- calculated by averaging the actions of the agent in each corresponding signal value across the episode.
% The shaded area represents one standard deviation.

We can see a clear temporal convention emerging: at signal value $g_1$ (first column), only agents $n_1$ and $n_3$ harvest (first and third row), at $g_2$, $n_2$ and $n_4$ harvest, at $g_3$, $n_1$ and $n_3$ harvest, and, finally, at $g_4$, $n_2$ and $n_4$. Contrary to that, in a sustainable joint strategy without the use of the signal, every agent harvests at every time-step with an average (across all agents) effort of $\approx 40\%$ (for the same setting of $N=4$ and $M_s=0.5$). \emph{Having all agents harvesting at every time-step makes coordination increasingly harder as we increase the population size}, mainly due to the non-stationarity of the environment (high variance) and the global exploration problem.

\subsubsection{\textbf{Access Rate}} \label{Access Rate}

In order to facilitate a systematic analysis of the accessing patterns, we discretized the agents into three bins: agents harvesting with effort $\epsilon \in [0 - 0.33)$ (`idle'), $[0.33 - 0.66)$ (`moderate'), and $[0.66 - 1]$ (`active'). Then we counted the average number of agents in each bin at the first equilibrium stock multiplier ($M_s$) where a non-depleting strategy was achieved in each setting. Without a signal, either the majority of the agents are `moderate' harvesters (specifically $84\%$ for $N=8$ and $16$), or \emph{all} of them are `active' harvesters ($100\%$ for $N=32$ and $64$). With the signal, we have a clear separation into `idle' and `active': (`idle', `active') = $(61\%, 30\%)$, $(59\%, 28\%)$, $(38\%, 44\%)$, $(50\%, 40\%)$, for $N = 8$, $16$, $32$, and $64$, respectively\footnote{The setting with $N = 64$ was run with $r = 2$ in both cases (with and without the signal). See Section \ref{supplement: Simulation Results} for more information.}. It is apparent that with the signal the agents learn a temporal convention; \emph{only a minority is `active' per time-step}, allowing to maintaining a healthy stock and reach \emph{sustainable strategies of high social welfare}.

\subsubsection{\textbf{Fallowing}} \label{Fallowing}

A more interesting joint strategy can be seen in Fig. \ref{fig:strategies_n2} ($N=2$, $M_s=0.5$). In this setting, we have an increased number of available signals, specifically $G=\frac{3N}{2}=3$. We can see that agents harvest alternatingly in the first two signal values, and rest on the third (\emph{fallow period}), potentially to allow resources to replenish and consequently obtain higher rewards in the future due to a higher catchability coefficient. This also resembles the optimal (bang-bang) harvesting strategy of Theorem \ref{Th: Optimal Harvesting Strategy}.

\begin{figure}[t!]
	\centering
	\begin{subfigure}[t]{1\linewidth}
		\centering
		\begin{minipage}{0.999995\linewidth}
		\includegraphics[width = 1\linewidth, trim={0.5em 0em 0em 1.2em}, clip]{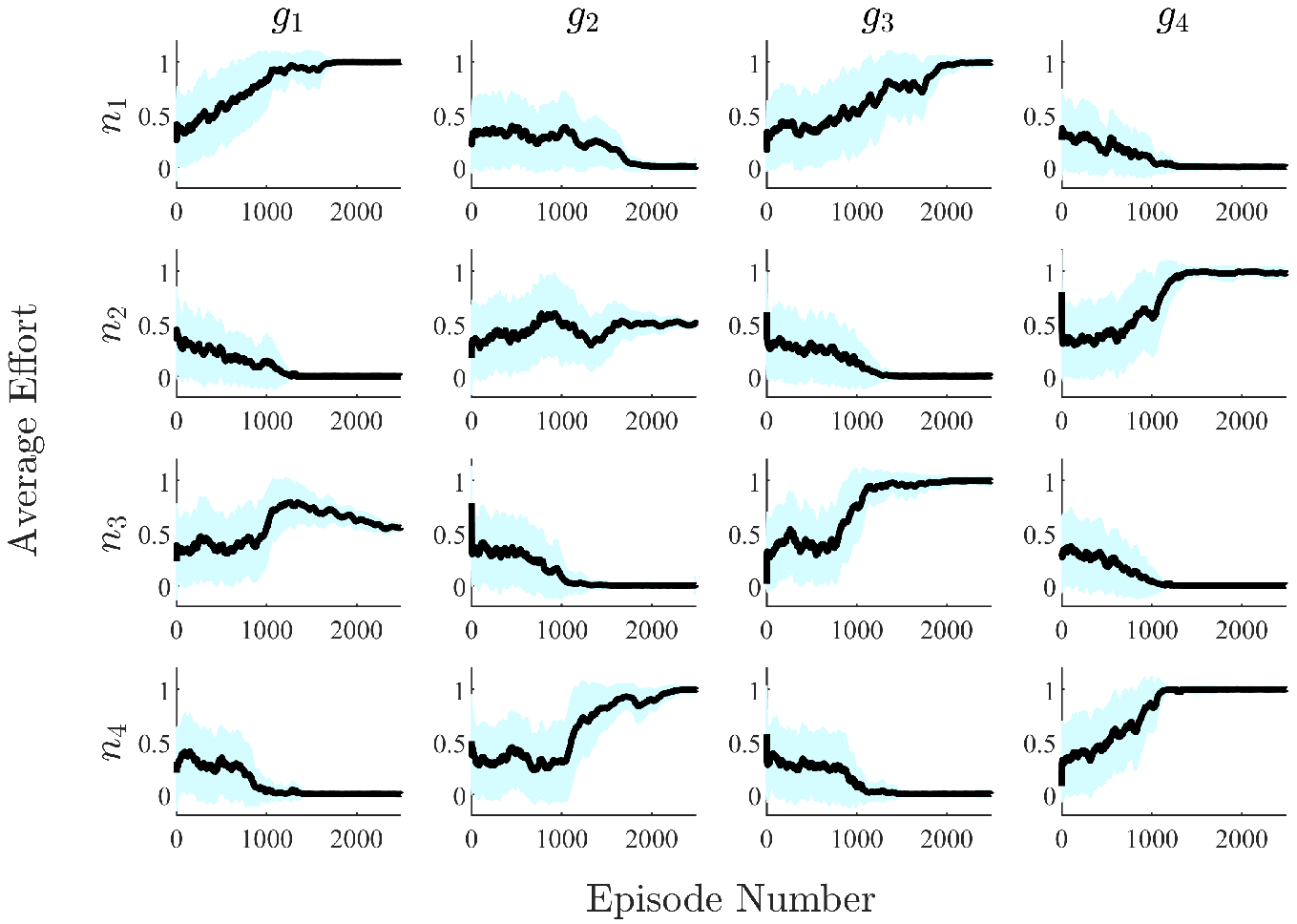}
		\end{minipage}\hfill
		\begin{minipage}{0.000005\linewidth}
		\caption{ }
		\label{fig:strategies_n4}
		\end{minipage}
	\end{subfigure}
	\hfill
	\begin{subfigure}[t]{1\linewidth}
		\centering
		\begin{minipage}{0.999995\linewidth}
		\includegraphics[width = 0.75\linewidth, trim={-1em 0em 0em 0em}, clip]{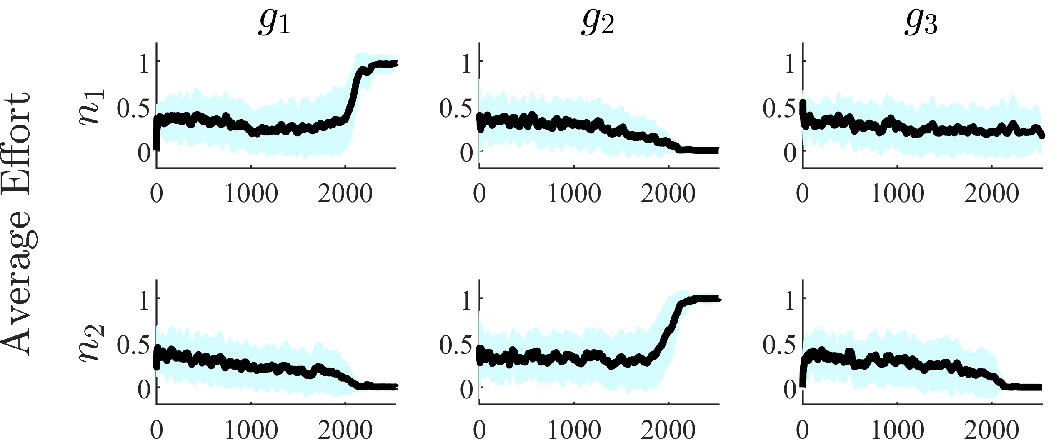}
		\end{minipage}\hfill
		\begin{minipage}{0.000005\linewidth}
		\caption{ }
		\label{fig:strategies_n2}
		\end{minipage}
	\end{subfigure}
	\hfill
	
	\caption{Evolution of the agents' strategies for each signal value, smoothed over 50 episodes. Fig. \ref{fig:strategies_n4} pertains to a population of $N=4$ and signal size $G=N=4$, while Fig. \ref{fig:strategies_n2} to a population of $N=2$ and signal size $G=\frac{3N}{2}=3$. In both cases the equilibrium stock multiplier is $M_s=0.5$. Each row represents an agent ($n_i$), while each column a signal value ($g_j$). Each line depicts the average effort the agent exerts on that specific signal value -- calculated by averaging the actions of the agent in each corresponding signal value across the episode. Shaded areas represent one standard deviation.}
	\Description{Evolution of the agents' strategies for each signal value.}
	\label{fig:strategies}
\end{figure}
% Begin: Example strategies

% \subsection{Fairness} \label{Results Fairness}

% Finally, we evaluated the fairness of the final allocation, to ensure that agents are not being exploited by the introduction of the signal. We used two of most established fairness metrics: the Jain index \cite{DBLP:journals/corr/cs-NI-9809099} and the Gini coefficient \cite{gini1912variabilita}. Both metrics showed that learning both with and without the signal results in fair allocations, with no significant change with the introduction of the signal. Numerical values can be found in the supplementary material.

\section{Conclusion} \label{Conclusion}

The challenge to cooperatively solve `the tragedy of the commons' remains as relevant now as when it was first introduced by \citeauthor{hardin1968tragedy} in \citeyear{hardin1968tragedy}. Sustainable development and avoidance of catastrophic scenarios in socio-ecological systems -- like the permanent depletion of resources, or the extinction of endangered species -- constitute critical open problems. To add to the challenge, real-world problems are inherently large in scale and of low observability. This amplifies traditional problems in multi-agent learning, such as the global exploration and the moving-target problem. Earlier work in common-pool resource appropriation utilized intrinsic or extrinsic incentives (e.g., reward or opponent shaping). Yet, such techniques need to be designed for the problem at hand and/or require communication or observability of states/actions, which is not always feasible (e.g., in commercial fisheries, the stock or harvesting efforts can not be directly observed). Humans on the other hand show a remarkable ability to self-organize and resolve common-pool resource dilemmas, often \emph{without any extrinsic incentive mechanism or communication}. Social conventions and the use of auxiliary environmental information constitute key mechanisms for the emergence of cooperation under low observability. In this paper, we demonstrate that utilizing such environmental signals -- which are amply available -- is a simple, yet powerful and robust technique, to foster cooperation in large-scale, low observability, and high-stakes environments. We are the first to tackle a realistic CPR appropriation scenario modeled on real-world commercial fisheries and under low observability. Our approach avoids permanent depletion in a wider (up to $300\%$) range of settings, while achieving higher social welfare (up to $3306\%$) and convergence speed (up to $53\%$).

% \clearpage

%%%%%%%%%%%%%%%%%%%%%%%%%%%%%%%%%%%%%%%%%%%%%%%%%%%%%%%%%%%%%%%%%%%%%%%%

%%% The next two lines define, first, the bibliography style to be 
%%% applied, and, second, the bibliography file to be used.

\bibliographystyle{ACM-Reference-Format} 
\bibliography{arXiv_fisheries_bibliography}

%%%%%%%%%%%%%%%%%%%%%%%%%%%%%%%%%%%%%%%%%%%%%%%%%%%%%%%%%%%%%%%%%%%%%%%%

\appendix

\section{Appendix}

\subsection{Contents}

In this appendix we include several details that have been omitted from the main text. In particular:

\begin{itemize}
	\item[-] In Section \ref{supplement: Proof of Theorem 2.1}, we prove Theorem \ref{Th: Optimal Harvesting Strategy}.
	\item[-] In Section \ref{supplement: Growth Rate}, we investigate the range of feasible values for the growth rate, $r$.
	\item[-] In Section \ref{supplement: Modeling Details}, we provide details on the agent architecture, the introduced signal, the CIC metric, and the fairness indices.
	\item[-] In Section \ref{supplement: Simulation Results}, we provide a thorough account of the simulation results.
\end{itemize}

% Begin: Model Dynamics
% trim={<left> <lower> <right> <upper>}
\begin{figure*}[t!]
	\centering
	\begin{subfigure}[t]{0.24\linewidth}
		\centering
		\includegraphics[width = 1\linewidth, trim={0em 0em 0em 0em}, clip]{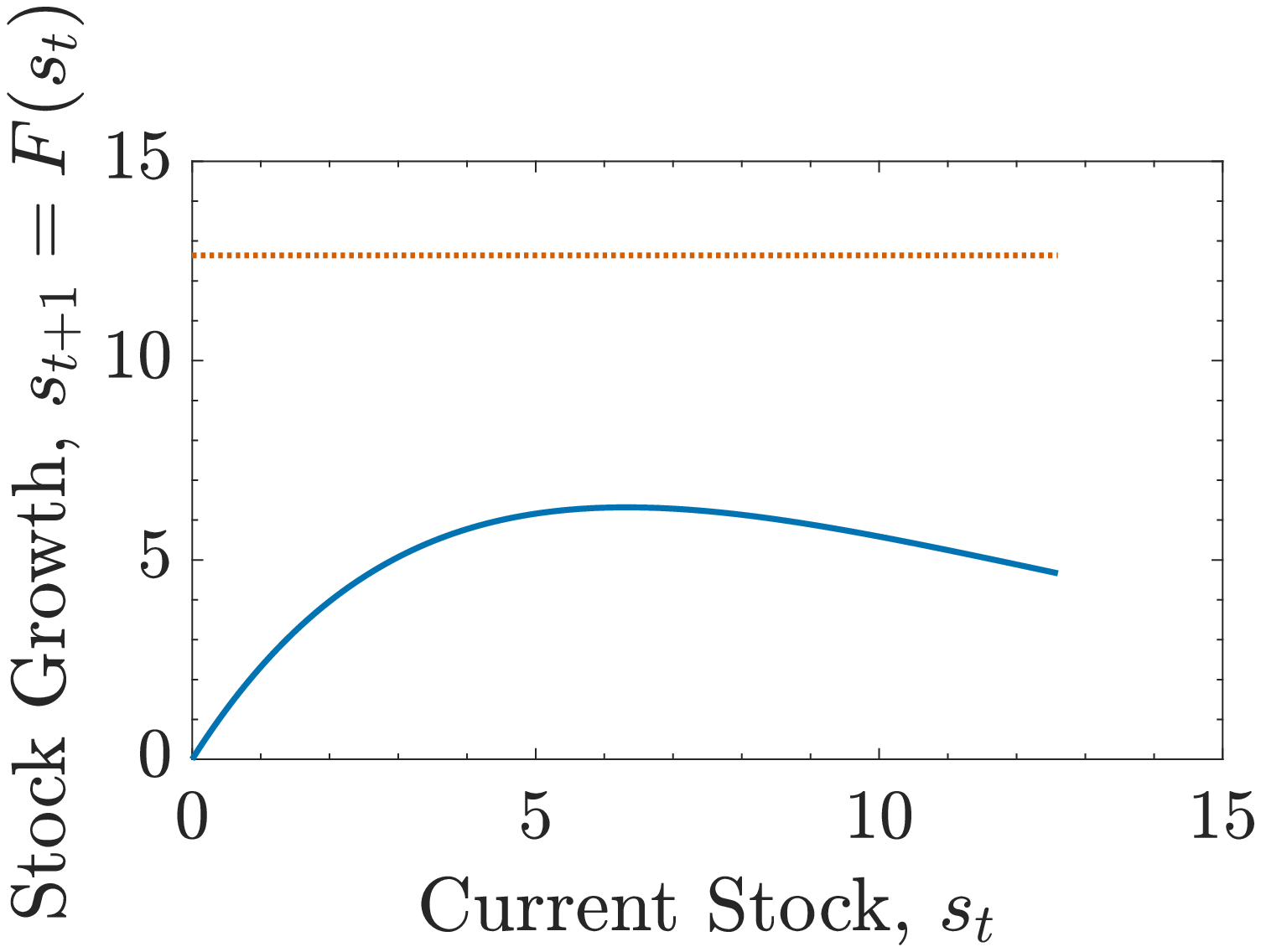}
		\caption{$r=1$}
		\label{fig:model_dynamics_r_1}
	\end{subfigure}
	\hfill
	\begin{subfigure}[t]{0.22\linewidth}
		\centering
		\includegraphics[width = 1\linewidth, trim={4em 0em 0em 0em}, clip]{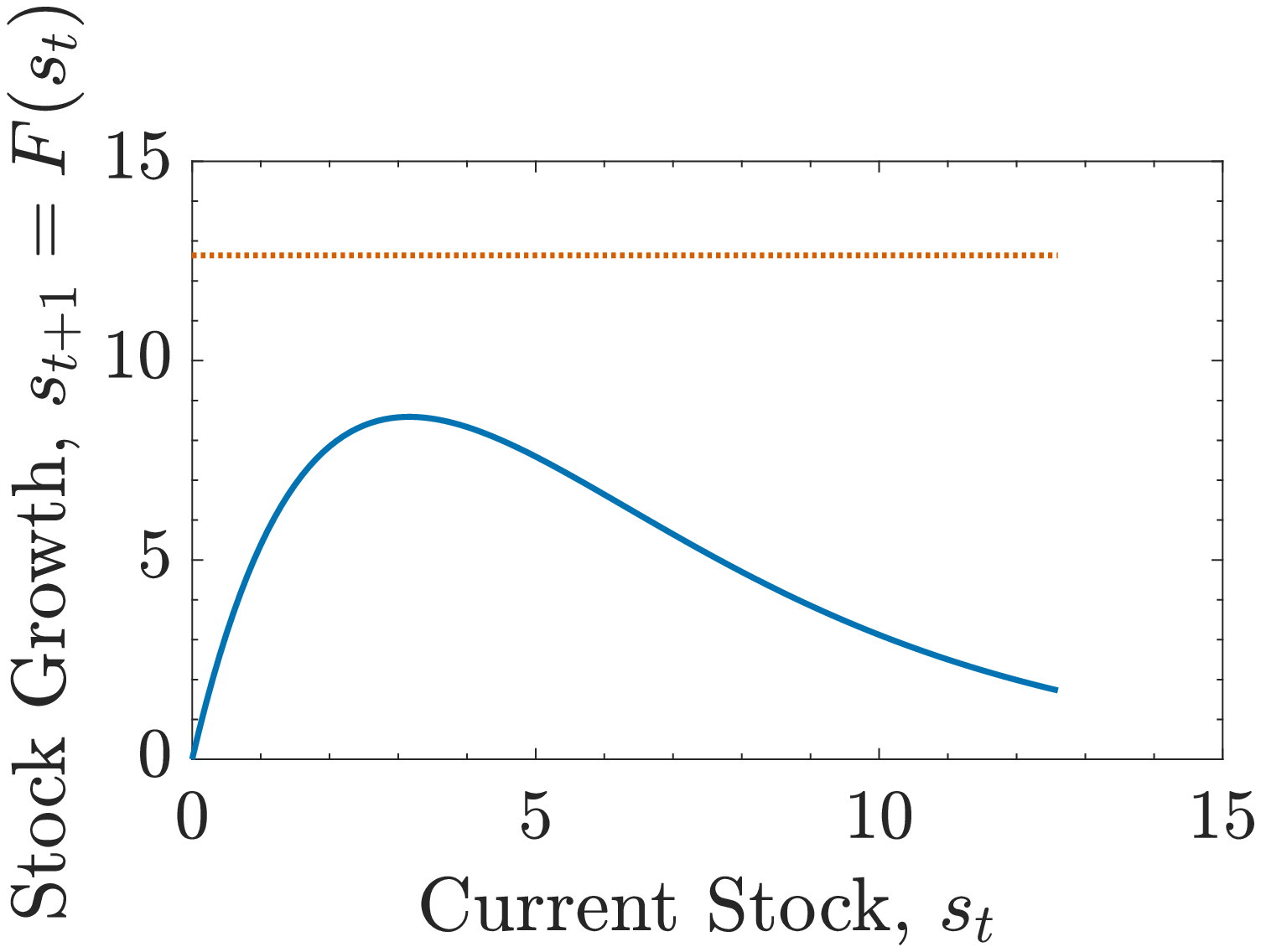}
		\caption{$r=2$}
		\label{fig:model_dynamics_r_2}
	\end{subfigure}
	\hfill
	\begin{subfigure}[t]{0.22\linewidth}
		\centering
		\includegraphics[width = 1\linewidth, trim={4em 0em 0em 0em}, clip]{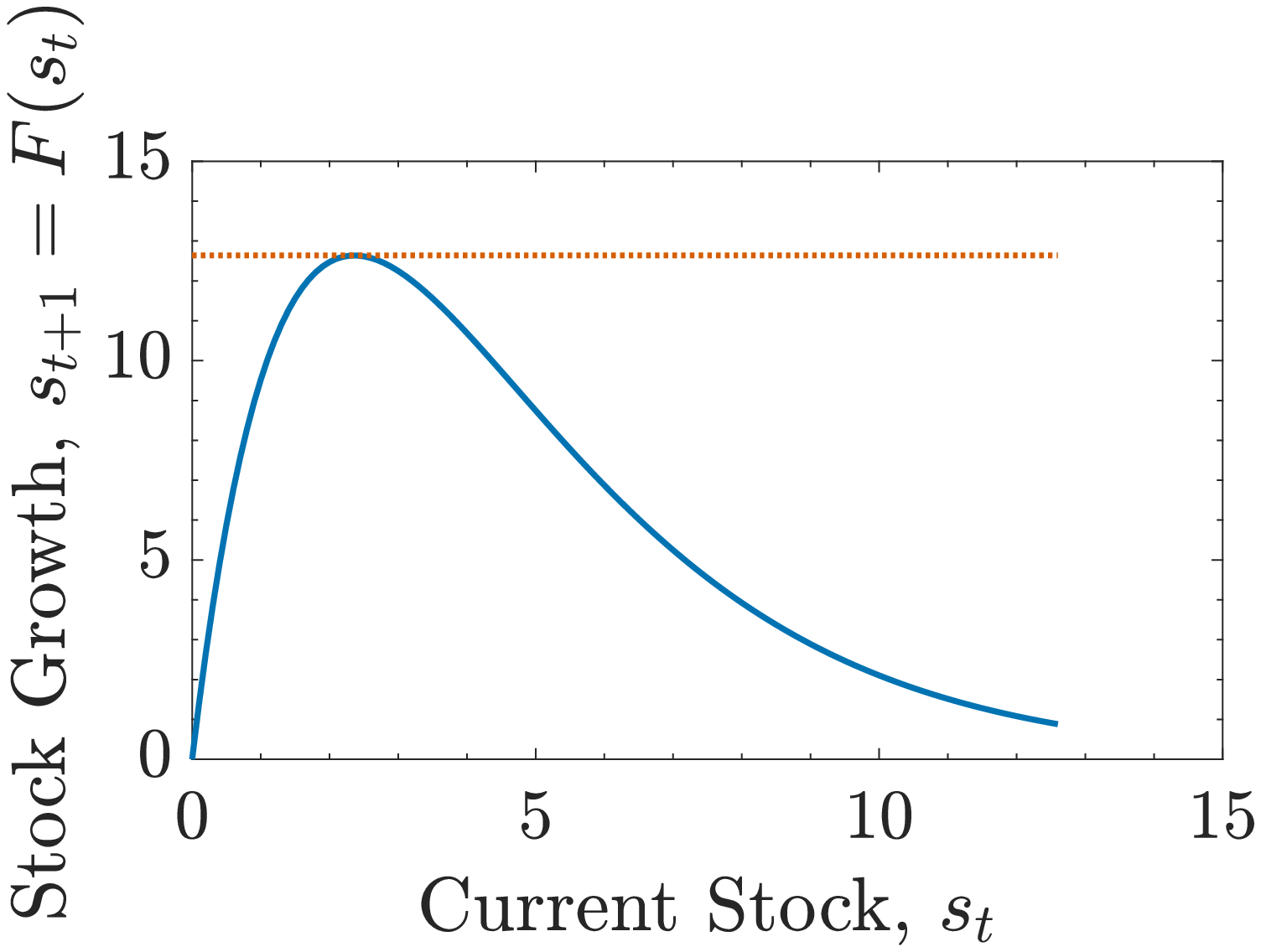}
		\caption{$r = -W_{-1}(-\frac{1}{2 e}) \approx 2.678$}
		\label{fig:model_dynamics_r_2p678}
	\end{subfigure}
	\hfill
	\begin{subfigure}[t]{0.215\linewidth}
		\centering
		\includegraphics[width = 1\linewidth, trim={3.5em 0em 0em 0em}, clip]{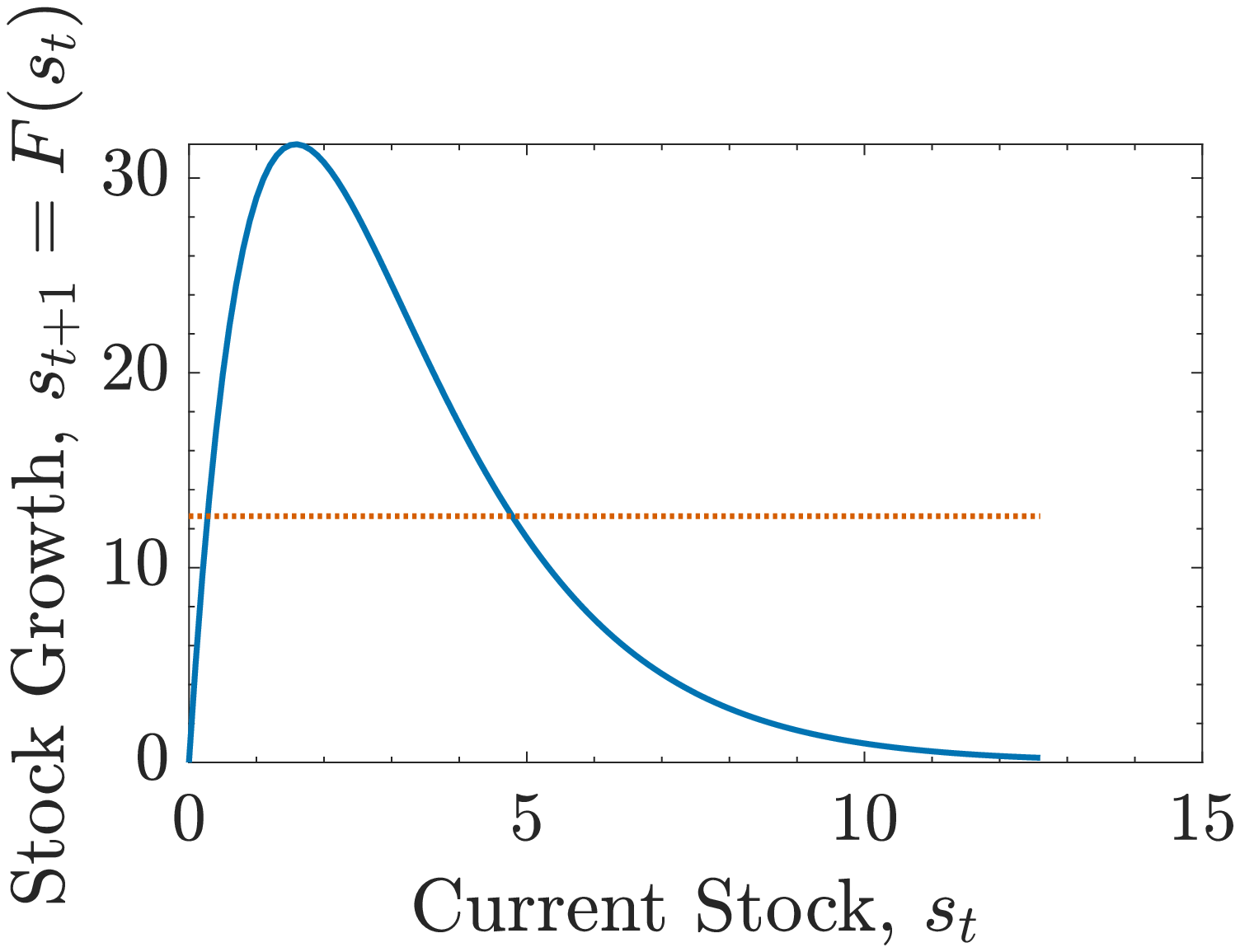}
		\caption{$r=4$}
		\label{fig:model_dynamics_r_4}
	\end{subfigure}
	\hfill
	
	\caption{Plot of the spawner-recruit function $F(\cdot)$ for various growth rates: $r = 1$, $2$, $-W_{-1}(-\frac{1}{2 e}) \approx 2.678$, and $4$ (Fig. \ref{fig:model_dynamics_r_1}, \ref{fig:model_dynamics_r_2}, \ref{fig:model_dynamics_r_2p678}, and \ref{fig:model_dynamics_r_4}, respectively). The $x$-axis denotes the current stock level ($s_t$), while the $y$-axis depicts the stock level on the next time-step, assuming no harvest (i.e., $s_{t+1} = F(s_t)$). The dashed line indicates a stock level equal to $2S_{eq}$.}
	\Description{Plot of the spawner-recruit function for various growth rates.}
	\label{fig:model_dynamics}
\end{figure*}
% End: Model Dynamics

\section{Proof of Theorem 2.1} \label{supplement: Proof of Theorem 2.1}

For completeness, we re-state the control problem and Theorem \ref{Th: Optimal Harvesting Strategy}. We want to find a piecewise continuous control $E_t$, so as to maximize the total revenue for a given episode duration $T$ (Eq. \ref{Eq: Cumulative Revenue proof}, where $U_t(E_t)$ is the cumulative revenue at time-step $t$, given by Eq. \ref{Eq: Utility Expression}). The maximization problem can be solved using Optimal Control Theory \cite{ding2010introduction_optimalcontrol,lenhart2007optimal}.

\begin{align}
\begin{split} \label{Eq: Cumulative Revenue proof}
	\underset{E_t}{\max}& \underset{t = 0}{\overset{T}{\sum}} U_{t}(E_{t}) \\
	\text{subject to } s_{t+1} &= F(s_t - H(E_t,s_t))
\end{split}
\end{align}

\begin{equation} \label{Eq: Utility Expression}
	U_t(E_t) = p_t H(E_t,s_t) - c_t
\end{equation}

The optimal\footnote{`Optimal' is used in a technical sense, as the strategy that maximizes the revenue subject to the model equations, and it does not carry any moralistic implications.} control is given by the following theorem:

\begin{customthm}{2.1}
	The optimal control variables $E^*_t$ that solves the maximization problem of Eq. \ref{Eq: Cumulative Revenue proof} given the model dynamics described in Section \ref{Fishery Model} is given by the following equation, where $\lambda_t$ are the adjoint variables of the Hamiltonians:

	\begin{equation*} \label{Eq: theorem optimal lambda proof}
		E^*_{t+1} =
		\begin{cases}
			E_{max} ,& \text{if } (p_{t+1}-\lambda_{t+1})q(F(s_t - H(E_t, s_t))) \geq 0 \\
			0  ,& \text{if } (p_{t+1}-\lambda_{t+1})q(F(s_t - H(E_t, s_t))) < 0 \\
		\end{cases}
	\end{equation*}
\end{customthm}

\begin{proof}
In order to decouple the state ($s_t$, the current resource stock) and the control ($E_t$)\footnote{In accordance to the literature on Optimal Control Theory \cite{lenhart2007optimal}, `state' in the context of the proof refers to the variable describing the the behavior of the underlying dynamical system, and `control' refers to the input function used to steer the state of the system.} and simplify the calculations, we resort to a change of variables. We define the new state $w_t$ as the remaining stock after harvest at time-step $t$:

\begin{equation*} \label{Eq: change of state variable - 1}
	w_t \triangleq s_t - H(E_t, s_t)
\end{equation*}

Therefore,

\begin{equation} \label{Eq: change of state variable - 2}
	s_{t+1}=w_{t+1}+H(E_{t+1},s_{t+1})
\end{equation}

and

\begin{equation} \label{Eq: change of state variable - 2.5}
	s_{t+1}=F(s_t - H(E_t,s_t))=F(w_t)
\end{equation}

Using Eq. \ref{Eq: change of state variable - 2} and \ref{Eq: change of state variable - 2.5}, we can write the new state equation as:

\begin{equation} \label{Eq: change of state variable - 3}
	w_{t+1}=F(w_t)-H(E_{t+1},F(w_t))
\end{equation}

In the current form of the state equation (Eq. \ref{Eq: change of state variable - 3}), the harvested resources appear outside the nonlinear growth function $F(.)$, making the following analysis significantly simpler.

Under optimal control, the resource will not get depleted before the end of the horizon $T$, thus $q(s_t)E_t\leq s_t$, $\forall t < T$\footnote{Let us assume this is not the case and the optimal strategy would deplete the stock at certain time-step $T_{dep}$. That means that rewards are $0$ and the optimal $E_t$ is arbitrary for $t>T_{dep}$. Using the modified equation for the total harvest (Eq. \ref{Eq: Total harvest modified}), we allow $H(E_t,s_t)>s_t$ or $s_t-H(E_t,s_t)<0$. This would lead to $s_{t+1}=F(s_t-H(E_t,s_t))<0$ $\forall t>T_{dep}$, i.e., the stock would become negative. In such a case, any positive effort would decrease the revenue (since it would result in a negative harvest), thus the optimal strategy would be to set $E_t=0$. Thus, using the modified equation for the total harvest (Eq. \ref{Eq: Total harvest modified}), does not change the optimal solution.}. We can rewrite the total harvest as:

\begin{equation} \label{Eq: Total harvest modified}
	H(E_{t+1},F(w_t))=q(F(w_t))E_{t+1}
\end{equation}

The state equation (Eq. \ref{Eq: change of state variable - 3}) becomes:

\begin{equation*} \label{Eq: change of state variable - 4}
	w_{t+1}=F(w_t)-q(F(w_t))E_{t+1}
\end{equation*}

and the optimization problem:

\begin{equation*} \label{Eq: Optimization Problem 2 - Objective}
	max \sum_{t=-1}^{T-1}(p_{t+1}q(F(w_t))E_{t+1}-c_{t+1})
\end{equation*}

\begin{equation*} \label{Eq: Optimization Problem 2 - State evolution}
	subject\ to\ w_{t+1} = F(w_t)-q(F(w_t))E_{t+1}
\end{equation*}
Let $w_{-1}=s_0=S_{eq}$. Solving the optimization problem is equivalent to finding the control that optimizes the Hamiltonians \cite{lenhart2007optimal,ding2010introduction_optimalcontrol}. Let $\lambda = (\lambda_{-1}, \lambda_{0}, \dots, \lambda_{T-1})$ denote the adjoint function. The Hamiltonian at time-step $t$ is given by:

\begin{equation} \label{Eq: Hamiltonian}
\begin{split}
	\textbf{H}_t & = p_{t+1}q(F(w_t))E_{t+1}-c_{t+1}+\lambda_{t+1}(F(w_t)-q(F(w_t))E_{t+1}) \\
	& = (p_{t+1}-\lambda_{t+1})q(F(w_t))E_{t+1}-c_{t+1}+\lambda_{t+1}F(w_t)
\end{split}
\end{equation}

The adjoint equations are given by \cite{ding2010introduction_optimalcontrol}:

\begin{equation*} \label{Eq: Hamiltonian_conditions}
\begin{split}
	\lambda_t & = \frac{\partial \mathbf{H}_t}{\partial w_t} \\
	\lambda_T & = 0 \\
	\frac{\partial \mathbf{H}_t}{\partial u_t} & = 0\ \textup{at}\ u_t=u^*_t
\end{split}
\end{equation*}

\noindent
where $u_t$ is the control input, which corresponds to $E_{t+1}$ in our formulation. The last condition corresponds to the maximization of the Hamiltonian $\textbf{H}_t$ in Eq. \ref{Eq: Hamiltonian} for all time-steps $t$ \cite[Chapter~23]{lenhart2007optimal}. In our case, Eq. \ref{Eq: Hamiltonian} is linear in $E_{t+1}$ with coefficient $(p_{t+1}-\lambda_{t+1})q(F(w_t))$. Therefore, the optimal sequence of $E_{t+1}$ that maximizes Eq. \ref{Eq: Hamiltonian} is given based on the sign of the coefficient:

 \begin{equation*} \label{Eq: Optimal harvest lambda - initial}
	E^*_{t+1} =
	\begin{cases}
		E_{max} ,& \text{if } (p_{t+1}-\lambda_{t+1})q(F(w_t)) \geq 0 \\
		0  ,& \text{if } (p_{t+1}-\lambda_{t+1})q(F(w_t)) < 0 \\
	\end{cases}
\end{equation*}

The optimal strategy therefore implements bang–bang controller, which oscillates based on adjoint variable values $\lambda_{t}$, current state $w_t$, and price $p_t$.

\end{proof}

\section{Growth Rate} \label{supplement: Growth Rate}

The growth rate, $r$, plays a significant role in the stability of the stock dynamics. A high growth rate -- and subsequently high population density -- can even lead to the extinction of the population due to the collapse in behavior from overcrowding (a phenomenon known as `behavioral sink' \cite{calhoun1962population,doi:10.1177/00359157730661P202}). This is reflected by the spawner-recruit function (Eq. \ref{Eq: Spawner-recruit function}) in our model. As depicted in Fig. \ref{fig:model_dynamics}, the higher the growth rate, $r$, the more skewed the stock dynamics. Specifically, Fig. \ref{fig:model_dynamics} plots the spawner-recruit function $F(\cdot)$ for various growth rates: $r = 1$, $2$, $-W_{-1}(-\frac{1}{2 e}) \approx 2.678$, and $4$ (Fig. \ref{fig:model_dynamics_r_1}, \ref{fig:model_dynamics_r_2}, \ref{fig:model_dynamics_r_2p678}, and \ref{fig:model_dynamics_r_4}, respectively). The $x$-axis denotes the current stock level ($s_t$), while the $y$-axis depicts the stock level on the next time-step, assuming no harvest (i.e., $s_{t+1} = F(s_t)$). The dashed line indicates a stock level equal to $2S_{eq}$. For a growth rate of $r = 4$ for example (Fig. \ref{fig:model_dynamics_r_4}), we have a highly skewed growth model that can lead to the depletion of the resource (high stock values ($s_t$) result to $s_{t+1} < \delta \rightarrow 0$, i.e., permanent depletion of the resource). For this reason, we want an unskewed growth model, specifically we want the stock to remain below two times the equilibrium stock point, i.e., $s_{t+1} \leq 2 S_{eq}$\footnote{This limit is imposed by the chosen parameters of the model equations.}. For this reason, we need to bound the growth rate according to the following theorem:

\begin{theorem}
	For a continuous resource governed by the dynamics of Section \ref{Fishery Model}, the stock value does not exceed the limit of $2 S_{eq}$, if $r \in [-W(-1/(2 e)), -W_{-1}(-1/(2 e))] \approx [0.232, 2.678]$, where $W_k(\cdot)$ is the Lambert $W$ function.
\end{theorem}

\begin{proof}
	Let $x \triangleq s_t \leq 2 S_{eq}$ for a time-step $t$. We want $s_{t+1} \leq 2 S_{eq}$, thus we need to bound the maximum value of the spawner-recruit function, $F(x)$ (Eq. \ref{Eq: Spawner-recruit function}).

	Taking the derivative:

	\begin{equation*}
		\frac{\partial}{\partial x} F(x) = e^{r \left( 1 - \frac{x}{S_{eq}} \right)} \left( 1 - \frac{r x}{S_{eq}} \right)
	\end{equation*}
	We have that:

	\begin{equation*}
		\frac{\partial}{\partial x} F(x) = 0 \Rightarrow x = \frac{S_{eq}}{r}, r \neq 0 \text{ and } S_{eq} \neq 0
	\end{equation*}
	Thus the maximum value is:

	\begin{equation*}
		F\left(\frac{S_{eq}}{r}\right) = \frac{S_{eq}}{r} e^{(r - 1)}
	\end{equation*}
	We want to bound the maximum value: 

	\begin{equation*}
	\begin{split}
		F\left(\frac{S_{eq}}{r}\right) \leq 2 S_{eq} &\Rightarrow \frac{e^{(r - 1)}}{r} \leq 2 \Rightarrow e^r - 2er \leq 0 \\
		&\Rightarrow -W\left(-\frac{1}{2 e}\right) \leq r \leq -W_{-1}\left(-\frac{1}{2 e}\right)
	\end{split}
	\end{equation*}

	\noindent
	where $W_k(\cdot)$ is the Lambert $W$ function. $-W(-1/(2 e)) \approx 0.232$ and $-W_{-1}(-1/(2 e)) \approx 2.678$.
\end{proof}

\section{Modeling Details} \label{supplement: Modeling Details}

\subsection{Agent Architecture Details}

Recent work has demonstrated that code-level optimizations play an important role in performance, both in terms of achieved reward and underlying algorithmic behavior \cite{Engstrom2020Implementation}. To minimize those sources of stochasticity -- and given that the focus of this work is in the performance of the introduced technique and not of the training algorithm -- we opted to use RLlib\footnote{RLlib (\url{https://docs.ray.io/en/latest/rllib.html}) is an open-source library on top of Ray (\url{https://docs.ray.io/en/latest/index.html}) for Multi-Agent Deep Reinforcement Learning \cite{Liang2017Ray}.} as our implementation framework. Each agent uses a two-layer ($64$ neurons each) feed-forward neural network for the policy approximation. The policies are trained using the Proximal Policy Optimization (PPO) algorithm \cite{DBLP:journals/corr/SchulmanWDRK17}. All the hyper-parameters were left to the default values specified in Ray and RLlib\footnote{See \url{https://docs.ray.io/en/latest/rllib-algorithms.html\#ppo}.}. For completeness, Table \ref{tab: List of hyper-parameters} presents a list of the most relevant of them.

\begin{table}[t!]
\centering
\caption{List of hyper-parameters.}
\label{tab: List of hyper-parameters}
\begin{tabular}{@{}lc@{}}
\toprule
\textbf{Parameter}                & \textbf{Value} \\ \midrule
Learning Rate ($\alpha$)          & 0.0001         \\
% Minibatch Size                    & 128            \\
Clipping Parameter                & 0.3            \\
Value Function Clipping Parameter & 10.0           \\
KL Target                         & 0.01           \\
Discount Factor ($\gamma$)        & 0.99           \\
GAE Parameter Lambda              & 1.0            \\
Value Function Loss Coefficient   & 1.0            \\
Entropy Coefficient               & 0.0            \\ \hline   
\end{tabular}%
\end{table}

\subsection{Introduced Signal: Implementation Details}

The introduced signal was encoded as a $G$-dimensional one-hot vector of fixed size, in which the high bit is shifted periodically. In particular, its value at index $i$ at time-step $t$ is given by:

\begin{equation} \label{Eq: Signal}
	\begin{cases}
		1 & \text{if } mod(t-t_{init},\ G)= i \\
		0 & \text{otherwise} \\
	\end{cases}
\end{equation}
\noindent
where $t_{init}$ is the random offset determined at the beginning of each episode in order to avoid learning any bias towards certain signal values.

\subsection{Causal Influence of Communication (CIC): Implementation Details} \label{supplement: Causal Influence of Communication}

The Causal Influence of Communication (CIC) \cite{lowe2019pitfalls} estimates the mutual information between the signal and the agent's action. The mutual information between two random variables $\tilde{X}$ and $\tilde{Y}$ is defined as the reduction of uncertainty (measured in terms of entropy $H_S(\cdot)$) in the value of $\tilde{X}$ with the observation of $\tilde{Y}$:

\begin{equation*}
	I(\tilde{X},\tilde{Y})=
	I(\tilde{Y},\tilde{X})=
	H_S(\tilde{X})-H_S(\tilde{X}|\tilde{Y})=
	E \left \{  
		log \left (
			\frac{
			P_{\tilde{X},\tilde{Y}}(\tilde{x},\tilde{y})}
			{P_{\tilde{X}}(\tilde{x})P_{\tilde{Y}}(\tilde{y})} 
		\right ) 
	\right \}
\end{equation*}

The pseudo-code for calculating the CIC for a single agent is presented in Alg. \ref{algo: CIC}. Note that the CIC implementation in \cite{lowe2019pitfalls} considers a multi-dimensional, one-hot, discrete action space with accessible probabilities for every action, while in our case we have a single, continuous action (specifically, the effort $\epsilon_{n,t}$). To solve this problem, we discretize our action space into $N_{bins}$ intervals between minimum ($\mathcal{E}_{min} = 0$) and maximum ($\mathcal{E}_{max} = 1$) effort values, and each interval is assumed to correspond to a single discrete action. Let $a_i$ denote the event of an action $\epsilon_{n,t}$ belonging to interval $i$. To calculate the CIC value, we start by generating $N_{states}$ random `partial' states (i.e., without signal), $\sigma_{-g}$, which are then concatenated with each possible signal value to obtain a `complete' state, $\sigma = [ g_j, \sigma_{-g}]$ (Lines 5 - 7 of Alg. \ref{algo: CIC}). Then, we estimate the probability of an action given a signal value, $p(a_i|g_j)$, by generating $N_{samples}$ actions from our policy given the `complete' state ($\pi(\sigma)$), and normalizing the number of instances in which the action belongs to a particular bin with the total number of samples. The remaining aspects of the calculation are the same as in the original implementation. In our calculations we used $N_{states}=N_{samples}=100$.

\begin{algorithm}[t!]
\SetAlgoLined

\textbf{input:} Agent policy $\pi(\cdot)$

$p(g_j) = \frac{1}{G}$ for all possible signals
 
Discretize the action space $[0, \mathcal{E}_{max}]$ into $N_{bins}$ intervals.
 
\For{i=1 to $N_{states}$}{
	Generate a state without a signal $\sigma_{-g}$ randomly
	
	\For{all possible signals $g_j$}{
		Generate agent observation $\sigma = [ g_j, \sigma_{-g}]$
		
		Estimate $p(a_i|g_j)$ by sampling $N_{samples}$ actions from $\pi(\sigma)$
		
		$p(a_i,g_j) = p(a_i|g_j)p(g_j)$
	}
	$p(a_i) = \sum_{j} p(a,g_j)$
	
	$CIC += \frac{1}{N_{states}}\sum_{a_i|p(a_i)\neq 0,\ g_j}^{}p(a,g_j)log(\frac{p(a,g_j)}{p(a)p(g_j)})$
}
\caption{CIC Implementation (based on \cite{lowe2019pitfalls})}
\label{algo: CIC}
\end{algorithm}

\subsection{Fairness Metrics}

We also evaluated the fairness of the final allocation, to ensure that agents are not being exploited by the introduction of the signal. We used two of most established fairness metrics: the Jain index \cite{DBLP:journals/corr/cs-NI-9809099} and the Gini coefficient \cite{gini1912variabilita}:

% \medskip
\textbf{(a)}  \emph{The Jain index} \cite{DBLP:journals/corr/cs-NI-9809099}: Widely used in network engineering to determine whether users or applications receive a fair share of system resources. It exhibits a lot of desirable properties such as: population size independence, continuity, scale and metric independence, and boundedness. For an allocation of $N$ agents, such that the $n^{\text{th}}$ agent is alloted $x_n$, the Jain index is given by Eq. \ref{eq: Jain Index}. $\mathds{J}(\mathbf{x}) \in [0, 1]$. An allocation $\mathbf{x} = (x_1, \dots, x_N) ^\top$ is considered fair, iff $\mathds{J}(\mathbf{x}) = 1$.

\begin{equation} \label{eq: Jain Index}
	\mathds{J}(\mathbf{x}) = \frac{\left(\underset{n = 1}{\overset{N}{\sum}} x_n\right) ^ 2}{N \underset{n = 1}{\overset{N}{\sum}}  x_n ^ 2}
\end{equation}

% \medskip
\textbf{(b)} \emph{The Gini coefficient} \cite{gini1912variabilita}: One of the most commonly used measures of inequality by economists intended to represent the wealth distribution of a population of a nation. For an allocation game of $N$ agents, such that the $n^{\text{th}}$ agent is alloted $x_n$, the Gini coefficient is given by Eq. \ref{eq: Gini coefficient}. $\mathds{G}(\mathbf{x}) \geq 0$. A Gini coefficient of zero expresses perfect equality, i.e., an allocation is fair iff $\mathds{G}(\mathbf{x}) = 0$.

\begin{equation} \label{eq: Gini coefficient}
	\mathds{G}(\mathbf{x}) = \frac{\underset{n = 1}{\overset{N}{\sum}} \underset{n' = 1}{\overset{N}{\sum}} \left| x_n - x_{n'} \right|}{2 N \underset{n = 1}{\overset{N}{\sum}}  x_n}
\end{equation}

Both metrics showed that learning both with and without the signal results in fair allocations, with no significant change with the introduction of the signal (see Section \ref{supplement: Simulation Results}).

\section{Simulation Results in Detail} \label{supplement: Simulation Results}

In this section we provide numerical values of the simulation results presented in the main text. Specifically:

Tables (\ref{tab: Social Welfare} and \ref{tab: Social Welfare relative difference and p values}), (\ref{tab: Episode Length} and \ref{tab: Episode Length relative difference and p values}), (\ref{tab: Training Time} and \ref{tab: Training Time relative difference and p values}), (\ref{tab: Jain Index} and \ref{tab: Jain Index relative difference and p values}), and (\ref{tab: Gini Coefficient} and \ref{tab: Gini Coefficient relative difference and p values}) include the results (absolute values with and without the introduced signal, relative difference, and Student's T-test p-values) on social welfare, episode lengths (in time-steps), training time (in number of episodes), Jain index, and Gini coefficient, respectively.

Table \ref{tab: CIC values} presents the CIC values.

Tables \ref{tab: Varying Signal Size Social Welfare}, \ref{tab: Varying Signal Size Episode Length}, \ref{tab: Varying Signal Size Training Time}, \ref{tab: Varying Signal Size Jain Index}, and \ref{tab: Varying Signal Size Gini Coefficient} present the results on the aforementioned metrics for varying signal size, $G = \{ 1, \frac{N}{2}, 23, N, 41, \frac{3N}{2} \}$.

We also ran simulations with higher growth rate, specifically $r = 2$. The results can be found in Tables \ref{tab: Growth Rate 2}, and \ref{tab: Growth Rate 2 relative difference and p values}. Every resource has a natural upper limit on the size of the population it can sustain. Fig. \ref{fig:episode_lengths} shows that as the number of agents grow, we reach sustainable strategies at higher equilibrium stock multipliers ($M_s$). Thus, we expect that as we increase the growth rate, the effect of the signal will be even more pronounced in larger populations ($N$) -- which is corroborated by the aforementioned Tables. The environment's ability to sustain a population also affects the counts of `active' agents of Section \ref{Access Rate}. As we can see in Fig. \ref{fig:episode_lengths}, for a growth rate of $r = 1$, the resource is too scarce, thus, even with the addition of the signal, the first sustainable strategy is reached at $M_s = 0.9$. This is a high equilibrium stock multiplier, close to the limit of sustainable harvesting. As a result, the number of `active' agents is naturally really high because they do not need to harvest in turns. By increasing the growth rate to $r = 2$, we have an environment that can sustain larger populations. Therefore, the first strategy that does not result to an immediate depletion is reached much earlier, and we can observe the emergence of a temporal convention (see Table \ref{tab: Access Rate}).

Finally, Table \ref{tab: Access Rate} shows the average number of agents in each bin -- $\epsilon \in [0 - 0.33)$ (`idle'), $[0.33 - 0.66)$ (`moderate'), and $[0.66 - 1]$ (`active') -- starting from the first equilibrium stock multiplier ($M_s$) where a non-depleting strategy was achieved in each setting.

All reported results are the average values over 8 trials. Note also that, as was specified in Section \ref{Environment Settings}, $K = \frac{S^{\scriptscriptstyle N, r}_{\scriptscriptstyle LSH}}{N} = \frac{e^r \mathcal{E}_{max}}{2 (e^r - 1)}$. Therefore, in the settings where the growth rate is $r = 1$, $K \approx 0.79$, while in the settings where $r = 2$, $K \approx 0.58$.

\begin{table*}[t!]
\centering
\caption{Social Welfare\\
Results (averaged over 8 trials) for increasing population size ($N \in [2, 64]$), with ($G=N$) and without ($G=1$) the introduced signal, for environments of decreasing difficulty (increasing $S_{eq} \propto M_s$).}
\label{tab: Social Welfare}
\resizebox{\textwidth}{!}{%
\begin{tabular}{@{}r|cccccccccccc@{}}
\toprule
&
\multicolumn{2}{c}{$N=$2} &
\multicolumn{2}{c}{$N=4$} &
\multicolumn{2}{c}{$N=8$} &
\multicolumn{2}{c}{$N=16$} &
\multicolumn{2}{c}{$N=32$} &
\multicolumn{2}{c}{$N=64$} \\
&
$G=1$ &
$G=N$ &
$G=1$ &
$G=N$ &
$G=1$ &
$G=N$ &
$G=1$ &
$G=N$ &
$G=1$ &
$G=N$ &
$G=1$ &
$G=N$ \\ \midrule
\textbf{$M_s = 0.2$} & 0.36   & 0.36   & 0.76   & 0.76   & 1.50    & 1.51    & 2.92    & 2.92    & 5.65    & 5.65    & 10.96   & 10.96   \\
\textbf{$M_s = 0.3$} & 0.60   & 0.60   & 1.28   & 1.28   & 2.64    & 2.64    & 5.38    & 4.81    & 8.26    & 8.61    & 16.38   & 15.29   \\
\textbf{$M_s = 0.4$} & 3.68   & 43.46  & 2.30   & 2.89   & 4.70    & 160.05  & 5.77    & 130.37  & 10.66   & 10.73   & 20.45   & 21.10   \\
\textbf{$M_s = 0.5$} & 122.30 & 111.57 & 179.20 & 186.65 & 141.89  & 193.53  & 70.10   & 154.00  & 12.81   & 39.35   & 25.67   & 25.81   \\
\textbf{$M_s = 0.6$} & 143.23 & 135.73 & 206.14 & 232.03 & 174.52  & 215.93  & 176.30  & 157.74  & 15.20   & 130.34  & 30.37   & 33.12   \\
\textbf{$M_s = 0.7$} & 172.69 & 179.28 & 181.15 & 223.18 & 121.33  & 211.70  & 196.26  & 202.63  & 34.14   & 196.57  & 41.49   & 52.57   \\
\textbf{$M_s = 0.8$} & 176.33 & 199.11 & 200.64 & 275.35 & 200.10  & 291.44  & 247.90  & 304.74  & 86.95   & 316.08  & 62.28   & 80.49   \\
\textbf{$M_s = 0.9$} & 197.49 & 207.15 & 244.93 & 316.75 & 327.91  & 395.14  & 341.24  & 499.66  & 165.75  & 637.30  & 101.12  & 1517.30 \\
\textbf{$M_s = 1.0$} & 223.28 & 229.01 & 339.33 & 371.72 & 482.25  & 524.64  & 585.71  & 687.65  & 592.26  & 1717.98 & 297.34  & 5350.52 \\
\textbf{$M_s = 1.1$} & 252.41 & 250.56 & 418.38 & 429.24 & 674.72  & 707.17  & 1315.68 & 1348.68 & 2611.96 & 3415.66 & 5225.12 & 8317.41 \\
\textbf{$M_s = 1.2$} & 280.61 & 278.61 & 523.86 & 532.54 & 1046.36 & 1051.10 & 2091.35 & 2120.13 & 4180.33 & 4784.53 & 8367.87 & 9900.67 \\ \hline
\end{tabular}
} \bigskip\bigskip
\end{table*}

\begin{table*}[t!]
\centering
\caption{Social Welfare (Relative difference \& p-values)\\
(i) Relative difference in Social Welfare when signal of cardinality $G=N$ is introduced ($(\text{Result}_{G=N} - \text{Result}_{G=1}) / \text{Result}_{G=1}$, where $\text{Result}_{G=X}$ denotes the achieved result using a signal of cardinality $X$), and \\
(ii) Student's T-test p-values, \\
for varying population size ($N \in [2, 64]$) and environments of decreasing difficulty (increasing $S_{eq} \propto M_s$).}
\label{tab: Social Welfare relative difference and p values}
\resizebox{\textwidth}{!}{%
\begin{tabular}{@{}r|cccccccccccc@{}}
\toprule
&
\multicolumn{2}{c}{$N=$2} &
\multicolumn{2}{c}{$N=4$} &
\multicolumn{2}{c}{$N=8$} &
\multicolumn{2}{c}{$N=16$} &
\multicolumn{2}{c}{$N=32$} &
\multicolumn{2}{c}{$N=64$} \\
&
(\%) &
p 	 &
(\%) &
p 	 &
(\%) &
p 	 &
(\%) &
p 	 &
(\%) &
p 	 &
(\%) &
p 	 \\ \midrule
\textbf{$M_s = 0.2$} & 0.0    & 0.87876 & 0.2  & 0.57508 & 0.2    & 0.47454 & -0.1   & 0.47064 & 0.1   & 0.36608 & 0.0    & 0.93297 \\
\textbf{$M_s = 0.3$} & -0.3   & 0.48043 & 0.2  & 0.57516 & 0.1    & 0.82779 & -10.5  & 0.02451 & 4.2   & 0.00000 & -6.6   & 0.00000 \\
\textbf{$M_s = 0.4$} & 1082.1 & 0.01102 & 25.8 & 0.16861 & 3305.9 & 0.00000 & 2158.1 & 0.00000 & 0.7   & 0.74501 & 3.2    & 0.01327 \\
\textbf{$M_s = 0.5$} & -8.8   & 0.00453 & 4.2  & 0.42501 & 36.4   & 0.00261 & 119.7  & 0.00566 & 207.2 & 0.14710 & 0.6    & 0.48706 \\
\textbf{$M_s = 0.6$} & -5.2   & 0.01401 & 12.6 & 0.01795 & 23.7   & 0.00033 & -10.5  & 0.48519 & 757.6 & 0.00001 & 9.1    & 0.00000 \\
\textbf{$M_s = 0.7$} & 3.8    & 0.00161 & 23.2 & 0.00000 & 74.5   & 0.00000 & 3.2    & 0.83672 & 475.7 & 0.00000 & 26.7   & 0.00000 \\
\textbf{$M_s = 0.8$} & 12.9   & 0.00032 & 37.2 & 0.00013 & 45.6   & 0.00105 & 22.9   & 0.00886 & 263.5 & 0.00000 & 29.2   & 0.00000 \\
\textbf{$M_s = 0.9$} & 4.9    & 0.06719 & 29.3 & 0.00008 & 20.5   & 0.12549 & 46.4   & 0.00355 & 284.5 & 0.00000 & 1400.5 & 0.00000 \\
\textbf{$M_s = 1.0$} & 2.6    & 0.03517 & 9.5  & 0.02099 & 8.8    & 0.01232 & 17.4   & 0.00014 & 190.1 & 0.00000 & 1699.5 & 0.00000 \\
\textbf{$M_s = 1.1$} & -0.7   & 0.41309 & 2.6  & 0.12172 & 4.8    & 0.01402 & 2.5    & 0.00002 & 30.8  & 0.00000 & 59.2   & 0.00000 \\
\textbf{$M_s = 1.2$} & -0.7   & 0.56215 & 1.7  & 0.00477 & 0.5    & 0.09644 & 1.4    & 0.00000 & 14.5  & 0.00000 & 18.3   & 0.00000 \\ \hline
\end{tabular}
}
\end{table*}

%%%%%%%%%%%%%%%%%%%%%%%%%%%%%%%%%%%%%%%%%%%%%%%%%%%%%%%%%%%%%%%
%%%%%%%%%%%%%%%%%%%%%%%%%%%%%%%%%%%%%%%%%%%%%%%%%%%%%%%%%%%%%%%

\begin{table*}[t!]
\centering
\caption{Episode Length (\#time-steps)\\
Results (averaged over 8 trials) for increasing population size ($N \in [2, 64]$), with ($G=N$) and without ($G=1$) the introduced signal, for environments of decreasing difficulty (increasing $S_{eq} \propto M_s$).}
\label{tab: Episode Length}
\resizebox{\textwidth}{!}{%
\begin{tabular}{@{}r|cccccccccccc@{}}
\toprule
&
\multicolumn{2}{c}{$N=$2} &
\multicolumn{2}{c}{$N=4$} &
\multicolumn{2}{c}{$N=8$} &
\multicolumn{2}{c}{$N=16$} &
\multicolumn{2}{c}{$N=32$} &
\multicolumn{2}{c}{$N=64$} \\
&
$G=1$ &
$G=N$ &
$G=1$ &
$G=N$ &
$G=1$ &
$G=N$ &
$G=1$ &
$G=N$ &
$G=1$ &
$G=N$ &
$G=1$ &
$G=N$ \\ \midrule
\textbf{$M_s = 0.2$} & 2.04   & 2.04   & 2.14   & 2.14   & 2.05   & 2.06   & 2.01   & 2.01   & 1.99   & 2.00   & 1.96   & 1.96   \\
\textbf{$M_s = 0.3$} & 2.75   & 2.72   & 2.83   & 2.82   & 3.63   & 3.66   & 5.17   & 3.61   & 1.98   & 2.67   & 2.20   & 1.19   \\
\textbf{$M_s = 0.4$} & 25.17  & 276.52 & 6.50   & 8.50   & 11.26  & 491.08 & 6.15   & 474.22 & 2.25   & 1.95   & 1.44   & 2.07   \\
\textbf{$M_s = 0.5$} & 499.79 & 499.94 & 463.01 & 500.00 & 481.93 & 500.00 & 240.27 & 496.08 & 1.90   & 119.85 & 2.62   & 2.32   \\
\textbf{$M_s = 0.6$} & 499.98 & 499.96 & 498.51 & 499.49 & 499.90 & 499.72 & 499.85 & 495.69 & 2.36   & 438.24 & 2.72   & 6.78   \\
\textbf{$M_s = 0.7$} & 500.00 & 500.00 & 499.24 & 500.00 & 499.49 & 500.00 & 500.00 & 500.00 & 70.78  & 500.00 & 10.00  & 31.33  \\
\textbf{$M_s = 0.8$} & 500.00 & 500.00 & 500.00 & 500.00 & 500.00 & 500.00 & 500.00 & 500.00 & 200.63 & 500.00 & 22.00  & 66.68  \\
\textbf{$M_s = 0.9$} & 500.00 & 500.00 & 500.00 & 500.00 & 500.00 & 500.00 & 500.00 & 500.00 & 275.00 & 500.00 & 53.00  & 500.00 \\
\textbf{$M_s = 1.0$} & 500.00 & 500.00 & 500.00 & 500.00 & 500.00 & 500.00 & 500.00 & 500.00 & 500.00 & 500.00 & 500.00 & 500.00 \\
\textbf{$M_s = 1.1$} & 500.00 & 500.00 & 500.00 & 500.00 & 500.00 & 500.00 & 500.00 & 500.00 & 500.00 & 500.00 & 500.00 & 500.00 \\
\textbf{$M_s = 1.2$} & 500.00 & 500.00 & 500.00 & 500.00 & 500.00 & 500.00 & 500.00 & 500.00 & 500.00 & 500.00 & 500.00 & 500.00 \\ \hline
\end{tabular}
} \bigskip\bigskip
\end{table*}

\begin{table*}[t!]
\centering
\caption{Episode Length (Relative difference \& p-values)\\
(i) Relative difference in Episode Length when signal of cardinality $G=N$ is introduced ($(\text{Result}_{G=N} - \text{Result}_{G=1}) / \text{Result}_{G=1}$, where $\text{Result}_{G=X}$ denotes the achieved result using a signal of cardinality $X$), and \\
(ii) Student's T-test p-values, \\
for varying population size ($N \in [2, 64]$) and environments of decreasing difficulty (increasing $S_{eq} \propto M_s$).\\
NaN values in the p-values column are due to having only a single data point; both cases (with and without the signal) have the same episode length in all the trials.}
\label{tab: Episode Length relative difference and p values}
\resizebox{\textwidth}{!}{%
\begin{tabular}{@{}r|cccccccccccc@{}}
\toprule
&
\multicolumn{2}{c}{$N=$2} &
\multicolumn{2}{c}{$N=4$} &
\multicolumn{2}{c}{$N=8$} &
\multicolumn{2}{c}{$N=16$} &
\multicolumn{2}{c}{$N=32$} &
\multicolumn{2}{c}{$N=64$} \\
&
(\%) &
p 	 &
(\%) &
p 	 &
(\%) &
p 	 &
(\%) &
p 	 &
(\%) &
p 	 &
(\%) &
p 	 \\ \midrule
\textbf{$M_s = 0.2$} & 0.0   & 0.99280 & 0.0  & 0.98122 & 0.7    & 0.43808 & -0.5   & 0.34488 & 0.6     & 0.32379 & 0.0   & 0.97624 \\
\textbf{$M_s = 0.3$} & -0.9  & 0.37060 & -0.2 & 0.83869 & 0.9    & 0.48086 & -30.2  & 0.02482 & 34.7    & 0.00000 & -46.0 & 0.00000 \\
\textbf{$M_s = 0.4$} & 998.6 & 0.00924 & 30.7 & 0.16075 & 4261.1 & 0.00000 & 7610.6 & 0.00000 & -13.3   & 0.52170 & 44.2  & 0.04255 \\
\textbf{$M_s = 0.5$} & 0.0   & 0.36593 & 8.0  & 0.04688 & 3.8    & 0.20338 & 106.5  & 0.01043 & 6209.4  & 0.14849 & -11.8 & 0.56968 \\
\textbf{$M_s = 0.6$} & 0.0   & 0.58271 & 0.2  & 0.36518 & 0.0    & 0.56152 & -0.8   & 0.02847 & 18503.0 & 0.00001 & 149.4 & 0.00001 \\
\textbf{$M_s = 0.7$} & 0.0   & NaN     & 0.2  & 0.13837 & 0.1    & 0.10198 & 0.0    & NaN     & 606.4   & 0.00001 & 213.3 & 0.00000 \\
\textbf{$M_s = 0.8$} & 0.0   & NaN     & 0.0  & NaN     & 0.0    & NaN     & 0.0    & NaN     & 149.2   & 0.00418 & 203.1 & 0.00000 \\
\textbf{$M_s = 0.9$} & 0.0   & NaN     & 0.0  & NaN     & 0.0    & NaN     & 0.0    & NaN     & 81.8    & 0.01919 & 843.4 & 0.00000 \\
\textbf{$M_s = 1.0$} & 0.0   & NaN     & 0.0  & NaN     & 0.0    & NaN     & 0.0    & NaN     & 0.0     & NaN     & 0.0   & NaN     \\
\textbf{$M_s = 1.1$} & 0.0   & NaN     & 0.0  & NaN     & 0.0    & NaN     & 0.0    & NaN     & 0.0     & NaN     & 0.0   & NaN     \\
\textbf{$M_s = 1.2$} & 0.0   & NaN     & 0.0  & NaN     & 0.0    & NaN     & 0.0    & NaN     & 0.0     & NaN     & 0.0   & NaN     \\ \hline
\end{tabular}
}
\end{table*}

%%%%%%%%%%%%%%%%%%%%%%%%%%%%%%%%%%%%%%%%%%%%%%%%%%%%%%%%%%%%%%%
%%%%%%%%%%%%%%%%%%%%%%%%%%%%%%%%%%%%%%%%%%%%%%%%%%%%%%%%%%%%%%%

\begin{table*}[t!]
\centering
\caption{Training Time (\#episodes)\\
Results (averaged over 8 trials) for increasing population size ($N \in [2, 64]$), with ($G=N$) and without ($G=1$) the introduced signal, for environments of decreasing difficulty (increasing $S_{eq} \propto M_s$).}
\label{tab: Training Time}
\resizebox{\textwidth}{!}{%
\begin{tabular}{@{}r|cccccccccccc@{}}
\toprule
&
\multicolumn{2}{c}{$N=$2} &
\multicolumn{2}{c}{$N=4$} &
\multicolumn{2}{c}{$N=8$} &
\multicolumn{2}{c}{$N=16$} &
\multicolumn{2}{c}{$N=32$} &
\multicolumn{2}{c}{$N=64$} \\
&
$G=1$ &
$G=N$ &
$G=1$ &
$G=N$ &
$G=1$ &
$G=N$ &
$G=1$ &
$G=N$ &
$G=1$ &
$G=N$ &
$G=1$ &
$G=N$ \\ \midrule
\textbf{$M_s = 0.2$} & 5000.00 & 5000.00 & 5000.00 & 5000.00 & 5000.00 & 5000.00 & 5000.00 & 5000.00 & 5000.00 & 5000.00 & 5000.00 & 5000.00 \\
\textbf{$M_s = 0.3$} & 5000.00 & 5000.00 & 5000.00 & 5000.00 & 5000.00 & 5000.00 & 5000.00 & 5000.00 & 5000.00 & 5000.00 & 5000.00 & 5000.00 \\
\textbf{$M_s = 0.4$} & 5000.00 & 5000.00 & 5000.00 & 5000.00 & 5000.00 & 4990.13 & 5000.00 & 4960.13 & 5000.00 & 5000.00 & 5000.00 & 5000.00 \\
\textbf{$M_s = 0.5$} & 4955.38 & 4340.63 & 4971.25 & 3408.75 & 5000.00 & 3389.50 & 5000.00 & 4847.88 & 5000.00 & 5000.00 & 5000.00 & 5000.00 \\
\textbf{$M_s = 0.6$} & 3037.38 & 2268.00 & 3822.75 & 2013.75 & 3903.00 & 3051.63 & 4583.13 & 2797.50 & 5000.00 & 4081.13 & 5000.00 & 5000.00 \\
\textbf{$M_s = 0.7$} & 1536.13 & 1008.13 & 2478.63 & 1328.00 & 3104.50 & 1880.00 & 3213.00 & 2680.00 & 4937.00 & 3026.00 & 5000.00 & 5000.00 \\
\textbf{$M_s = 0.8$} & 1004.00 & 1124.00 & 1968.00 & 1296.00 & 3088.00 & 2064.00 & 3778.00 & 2116.00 & 4887.00 & 2551.00 & 5000.00 & 5000.00 \\
\textbf{$M_s = 0.9$} & 1248.00 & 936.00  & 2168.00 & 1416.00 & 3025.00 & 1412.00 & 3810.00 & 2100.00 & 4536.00 & 2896.00 & 5000.00 & 2581.00 \\
\textbf{$M_s = 1.0$} & 864.00  & 776.00  & 1692.00 & 1252.00 & 2116.00 & 1612.00 & 2740.00 & 1788.00 & 3393.00 & 2334.00 & 1484.00 & 1965.00 \\
\textbf{$M_s = 1.1$} & 684.00  & 800.00  & 1028.00 & 1376.00 & 1284.00 & 1436.00 & 1164.00 & 1536.00 & 1218.00 & 1832.00 & 1193.00 & 1172.00 \\
\textbf{$M_s = 1.2$} & 840.00  & 780.00  & 1132.00 & 1040.00 & 1004.00 & 1212.00 & 932.00  & 1116.00 & 949.00  & 844.00  & 897.00  & 816.00  \\ \hline
\end{tabular}
} \bigskip\bigskip
\end{table*}

\begin{table*}[t!]
\centering
\caption{Training Time (Relative difference \& p-values)\\
(i) Relative difference in Training Time when signal of cardinality $G=N$ is introduced ($(\text{Result}_{G=N} - \text{Result}_{G=1}) / \text{Result}_{G=1}$, where $\text{Result}_{G=X}$ denotes the achieved result using a signal of cardinality $X$), and \\
(ii) Student's T-test p-values, \\
for varying population size ($N \in [2, 64]$) and environments of decreasing difficulty (increasing $S_{eq} \propto M_s$).}
\label{tab: Training Time relative difference and p values}
\resizebox{\textwidth}{!}{%
\begin{tabular}{@{}r|cccccccccccc@{}}
\toprule
&
\multicolumn{2}{c}{$N=$2} &
\multicolumn{2}{c}{$N=4$} &
\multicolumn{2}{c}{$N=8$} &
\multicolumn{2}{c}{$N=16$} &
\multicolumn{2}{c}{$N=32$} &
\multicolumn{2}{c}{$N=64$} \\
&
(\%) &
p 	 &
(\%) &
p 	 &
(\%) &
p 	 &
(\%) &
p 	 &
(\%) &
p 	 &
(\%) &
p 	 \\ \midrule
\textbf{$M_s = 0.2$} & 0.0   & 0.00000 & 0.0   & 0.00000 & 0.0   & 0.00000 & 0.0   & 0.00000 & 0.0   & 0.00000 & 0.0   & 0.00000 \\
\textbf{$M_s = 0.3$} & 0.0   & 0.00000 & 0.0   & 0.00000 & 0.0   & 0.00000 & 0.0   & 0.00000 & 0.0   & 0.00000 & 0.0   & 0.00000 \\
\textbf{$M_s = 0.4$} & 0.0   & 0.00000 & 0.0   & 0.00000 & -0.2  & 0.00000 & -0.8  & 0.00000 & 0.0   & 0.00000 & 0.0   & 0.00000 \\
\textbf{$M_s = 0.5$} & -12.4 & 0.00000 & -31.4 & 0.00000 & -32.2 & 0.00000 & -3.0  & 0.00000 & 0.0   & 0.00000 & 0.0   & 0.00000 \\
\textbf{$M_s = 0.6$} & -25.3 & 0.00000 & -47.3 & 0.00006 & -21.8 & 0.00013 & -39.0 & 0.00003 & -18.4 & 0.00000 & 0.0   & 0.00000 \\
\textbf{$M_s = 0.7$} & -34.4 & 0.00000 & -46.4 & 0.02193 & -39.4 & 0.00396 & -16.6 & 0.00666 & -38.7 & 0.00000 & 0.0   & 0.00000 \\
\textbf{$M_s = 0.8$} & 12.0  & 0.01448 & -34.1 & 0.00266 & -33.2 & 0.00228 & -44.0 & 0.00128 & -47.8 & 0.00000 & 0.0   & 0.00000 \\
\textbf{$M_s = 0.9$} & -25.0 & 0.00462 & -34.7 & 0.01569 & -53.3 & 0.00583 & -44.9 & 0.00221 & -36.2 & 0.00001 & -48.4 & 0.00000 \\
\textbf{$M_s = 1.0$} & -10.2 & 0.00137 & -26.0 & 0.02400 & -23.8 & 0.01508 & -34.7 & 0.03351 & -31.2 & 0.00211 & 32.4  & 0.00000 \\
\textbf{$M_s = 1.1$} & 17.0  & 0.00000 & 33.9  & 0.00024 & 11.8  & 0.00006 & 32.0  & 0.00000 & 50.4  & 0.00001 & -1.8  & 0.00000 \\
\textbf{$M_s = 1.2$} & -7.1  & 0.00847 & -8.1  & 0.00002 & 20.7  & 0.00000 & 19.7  & 0.00000 & -11.1 & 0.00000 & -9.0  & 0.00000 \\ \hline
\end{tabular}
}
\end{table*}

%%%%%%%%%%%%%%%%%%%%%%%%%%%%%%%%%%%%%%%%%%%%%%%%%%%%%%%%%%%%%%%
%%%%%%%%%%%%%%%%%%%%%%%%%%%%%%%%%%%%%%%%%%%%%%%%%%%%%%%%%%%%%%%

\begin{table*}[t!]
\centering
\caption{Jain Index (higher is fairer)\\
Results (averaged over 8 trials) for increasing population size ($N \in [2, 64]$), with ($G=N$) and without ($G=1$) the introduced signal, for environments of decreasing difficulty (increasing $S_{eq} \propto M_s$).}
\label{tab: Jain Index}
\resizebox{\textwidth}{!}{%
\begin{tabular}{@{}r|cccccccccccc@{}}
\toprule
&
\multicolumn{2}{c}{$N=$2} &
\multicolumn{2}{c}{$N=4$} &
\multicolumn{2}{c}{$N=8$} &
\multicolumn{2}{c}{$N=16$} &
\multicolumn{2}{c}{$N=32$} &
\multicolumn{2}{c}{$N=64$} \\
&
$G=1$ &
$G=N$ &
$G=1$ &
$G=N$ &
$G=1$ &
$G=N$ &
$G=1$ &
$G=N$ &
$G=1$ &
$G=N$ &
$G=1$ &
$G=N$ \\ \midrule
\textbf{$M_s = 0.2$} & 0.99984 & 0.99973 & 0.99976 & 0.99984 & 0.99963 & 0.99966 & 0.99967 & 0.99961 & 0.99943 & 0.99947 & 0.99932 & 0.99941 \\
\textbf{$M_s = 0.3$} & 0.99973 & 0.99991 & 0.99984 & 0.99981 & 0.99967 & 0.99974 & 0.99963 & 0.99959 & 0.99958 & 0.99930 & 0.99936 & 0.99942 \\
\textbf{$M_s = 0.4$} & 0.99210 & 0.99299 & 0.99910 & 0.99846 & 0.99119 & 0.98808 & 0.98691 & 0.99370 & 0.98603 & 0.99526 & 0.99063 & 0.99648 \\
\textbf{$M_s = 0.5$} & 0.98061 & 0.98251 & 0.97149 & 0.97895 & 0.97910 & 0.98978 & 0.98376 & 0.99416 & 0.98946 & 0.99694 & 0.99489 & 0.99773 \\
\textbf{$M_s = 0.6$} & 0.98252 & 0.99911 & 0.98507 & 0.99479 & 0.97847 & 0.99246 & 0.98147 & 0.99599 & 0.99898 & 0.99628 & 0.99935 & 0.99831 \\
\textbf{$M_s = 0.7$} & 0.99317 & 0.99996 & 0.98657 & 0.99239 & 0.98647 & 0.99458 & 0.98862 & 0.99601 & 0.99635 & 0.99660 & 1.00000 & 0.99787 \\
\textbf{$M_s = 0.8$} & 0.99705 & 0.99872 & 0.98081 & 0.99269 & 0.97903 & 0.99534 & 0.98679 & 0.99645 & 0.99023 & 0.99733 & 1.00000 & 0.99864 \\
\textbf{$M_s = 0.9$} & 0.99780 & 0.99535 & 0.97744 & 0.99304 & 0.97841 & 0.99581 & 0.98236 & 0.99704 & 0.99234 & 0.99810 & 0.99999 & 0.99912 \\
\textbf{$M_s = 1.0$} & 0.98674 & 0.99740 & 0.98892 & 0.99454 & 0.99094 & 0.99715 & 0.99441 & 0.99876 & 0.99578 & 0.99924 & 1.00000 & 0.99929 \\
\textbf{$M_s = 1.1$} & 0.99620 & 0.99498 & 0.99770 & 0.99737 & 0.99990 & 0.99976 & 0.99999 & 0.99999 & 1.00000 & 0.99945 & 0.99999 & 0.99909 \\
\textbf{$M_s = 1.2$} & 0.99855 & 0.99956 & 0.99998 & 0.99993 & 0.99999 & 0.99999 & 1.00000 & 0.99998 & 0.99999 & 0.99917 & 0.99999 & 0.99876 \\ \hline
\end{tabular}
} \bigskip\bigskip
\end{table*}

\begin{table*}[t!]
\centering
\caption{Jain Index (Relative difference \& p-values)\\
(i) Relative difference in Jain Index when signal of cardinality $G=N$ is introduced ($(\text{Result}_{G=N} - \text{Result}_{G=1}) / \text{Result}_{G=1}$, where $\text{Result}_{G=X}$ denotes the achieved result using a signal of cardinality $X$), and \\
(ii) Student's T-test p-values, \\
for varying population size ($N \in [2, 64]$) and environments of decreasing difficulty (increasing $S_{eq} \propto M_s$).}
\label{tab: Jain Index relative difference and p values}
\resizebox{\textwidth}{!}{%
\begin{tabular}{@{}r|cccccccccccc@{}}
\toprule
&
\multicolumn{2}{c}{$N=$2} &
\multicolumn{2}{c}{$N=4$} &
\multicolumn{2}{c}{$N=8$} &
\multicolumn{2}{c}{$N=16$} &
\multicolumn{2}{c}{$N=32$} &
\multicolumn{2}{c}{$N=64$} \\
&
(\%) &
p 	 &
(\%) &
p 	 &
(\%) &
p 	 &
(\%) &
p 	 &
(\%) &
p 	 &
(\%) &
p 	 \\ \midrule
\textbf{$M_s = 0.2$} & 0.0  & 0.62689 & 0.0  & 0.19181 & 0.0  & 0.73995 & 0.0 & 0.43584 & 0.0  & 0.67344 & 0.0  & 0.12180 \\
\textbf{$M_s = 0.3$} & 0.0  & 0.19431 & 0.0  & 0.74755 & 0.0  & 0.40216 & 0.0 & 0.68556 & 0.0  & 0.00316 & 0.0  & 0.24335 \\
\textbf{$M_s = 0.4$} & 0.1  & 0.85540 & -0.1 & 0.47459 & -0.3 & 0.31067 & 0.7 & 0.00312 & 0.9  & 0.00014 & 0.6  & 0.00000 \\
\textbf{$M_s = 0.5$} & 0.2  & 0.88279 & 0.8  & 0.27542 & 1.1  & 0.08226 & 1.1 & 0.00027 & 0.8  & 0.00105 & 0.3  & 0.00025 \\
\textbf{$M_s = 0.6$} & 1.7  & 0.05974 & 1.0  & 0.02806 & 1.4  & 0.00501 & 1.5 & 0.00055 & -0.3 & 0.00000 & -0.1 & 0.00000 \\
\textbf{$M_s = 0.7$} & 0.7  & 0.16061 & 0.6  & 0.25648 & 0.8  & 0.00610 & 0.7 & 0.00002 & 0.0  & 0.94743 & -0.2 & 0.00000 \\
\textbf{$M_s = 0.8$} & 0.2  & 0.33436 & 1.2  & 0.03393 & 1.7  & 0.00015 & 1.0 & 0.00049 & 0.7  & 0.16118 & -0.1 & 0.00000 \\
\textbf{$M_s = 0.9$} & -0.2 & 0.36223 & 1.6  & 0.00338 & 1.8  & 0.00011 & 1.5 & 0.00000 & 0.6  & 0.07503 & -0.1 & 0.00000 \\
\textbf{$M_s = 1.0$} & 1.1  & 0.07827 & 0.6  & 0.08720 & 0.6  & 0.00834 & 0.4 & 0.00017 & 0.3  & 0.00165 & -0.1 & 0.00000 \\
\textbf{$M_s = 1.1$} & -0.1 & 0.49041 & 0.0  & 0.64625 & 0.0  & 0.11982 & 0.0 & 0.18597 & -0.1 & 0.00006 & -0.1 & 0.00000 \\
\textbf{$M_s = 1.2$} & 0.1  & 0.05751 & 0.0  & 0.15660 & 0.0  & 0.38672 & 0.0 & 0.00003 & -0.1 & 0.00000 & -0.1 & 0.00000 \\ \hline
\end{tabular}
}
\end{table*}

%%%%%%%%%%%%%%%%%%%%%%%%%%%%%%%%%%%%%%%%%%%%%%%%%%%%%%%%%%%%%%%
%%%%%%%%%%%%%%%%%%%%%%%%%%%%%%%%%%%%%%%%%%%%%%%%%%%%%%%%%%%%%%%

\begin{table*}[t!]
\centering
\caption{Gini Coefficient (lower is fairer)\\
Results (averaged over 8 trials) for increasing population size ($N \in [2, 64]$), with ($G=N$) and without ($G=1$) the introduced signal, for environments of decreasing difficulty (increasing $S_{eq} \propto M_s$).}
\label{tab: Gini Coefficient}
\resizebox{\textwidth}{!}{%
\begin{tabular}{@{}r|cccccccccccc@{}}
\toprule
&
\multicolumn{2}{c}{$N=$2} &
\multicolumn{2}{c}{$N=4$} &
\multicolumn{2}{c}{$N=8$} &
\multicolumn{2}{c}{$N=16$} &
\multicolumn{2}{c}{$N=32$} &
\multicolumn{2}{c}{$N=64$} \\
&
$G=1$ &
$G=N$ &
$G=1$ &
$G=N$ &
$G=1$ &
$G=N$ &
$G=1$ &
$G=N$ &
$G=1$ &
$G=N$ &
$G=1$ &
$G=N$ \\ \midrule
\textbf{$M_s = 0.2$} & 0.00583 & 0.00551 & 0.00764 & 0.00663 & 0.01048 & 0.00956 & 0.00991 & 0.01087 & 0.01323 & 0.01282 & 0.01459 & 0.01356 \\
\textbf{$M_s = 0.3$} & 0.00674 & 0.00390 & 0.00644 & 0.00656 & 0.00983 & 0.00855 & 0.01040 & 0.01085 & 0.01135 & 0.01468 & 0.01416 & 0.01345 \\
\textbf{$M_s = 0.4$} & 0.03897 & 0.03083 & 0.01409 & 0.01667 & 0.05032 & 0.05879 & 0.06336 & 0.04329 & 0.06617 & 0.03841 & 0.05399 & 0.03328 \\
\textbf{$M_s = 0.5$} & 0.05571 & 0.05590 & 0.08869 & 0.07415 & 0.07528 & 0.05419 & 0.06839 & 0.04166 & 0.05646 & 0.03109 & 0.03890 & 0.02668 \\
\textbf{$M_s = 0.6$} & 0.05350 & 0.01113 & 0.06094 & 0.03547 & 0.07879 & 0.04711 & 0.07156 & 0.03444 & 0.01466 & 0.03405 & 0.01272 & 0.02298 \\
\textbf{$M_s = 0.7$} & 0.02754 & 0.00215 & 0.05777 & 0.04123 & 0.06279 & 0.03842 & 0.05894 & 0.03514 & 0.01252 & 0.03259 & 0.00038 & 0.02537 \\
\textbf{$M_s = 0.8$} & 0.02160 & 0.01446 & 0.07129 & 0.04442 & 0.07916 & 0.03539 & 0.06264 & 0.03255 & 0.03503 & 0.02884 & 0.00012 & 0.02059 \\
\textbf{$M_s = 0.9$} & 0.01901 & 0.02562 & 0.08016 & 0.03768 & 0.07998 & 0.03427 & 0.07229 & 0.02973 & 0.03275 & 0.02440 & 0.00027 & 0.01639 \\
\textbf{$M_s = 1.0$} & 0.04736 & 0.02190 & 0.05277 & 0.03522 & 0.04490 & 0.02877 & 0.03175 & 0.01899 & 0.01802 & 0.01524 & 0.00012 & 0.01480 \\
\textbf{$M_s = 1.1$} & 0.02741 & 0.03299 & 0.02376 & 0.02583 & 0.00392 & 0.00692 & 0.00119 & 0.00191 & 0.00035 & 0.01244 & 0.00038 & 0.01670 \\
\textbf{$M_s = 1.2$} & 0.01585 & 0.00951 & 0.00138 & 0.00358 & 0.00102 & 0.00154 & 0.00078 & 0.00235 & 0.00083 & 0.01604 & 0.00113 & 0.01958 \\ \hline
\end{tabular}
} \bigskip\bigskip
\end{table*}

\begin{table*}[t!]
\centering
\caption{Gini Coefficient (Relative difference \& p-values)\\
(i) Relative difference in Gini Coefficient when signal of cardinality $G=N$ is introduced ($(\text{Result}_{G=N} - \text{Result}_{G=1}) / \text{Result}_{G=1}$, where $\text{Result}_{G=X}$ denotes the achieved result using a signal of cardinality $X$), and \\
(ii) Student's T-test p-values, \\
for varying population size ($N \in [2, 64]$) and environments of decreasing difficulty (increasing $S_{eq} \propto M_s$).}
\label{tab: Gini Coefficient relative difference and p values}
\resizebox{\textwidth}{!}{%
\begin{tabular}{@{}r|cccccccccccc@{}}
\toprule
&
\multicolumn{2}{c}{$N=$2} &
\multicolumn{2}{c}{$N=4$} &
\multicolumn{2}{c}{$N=8$} &
\multicolumn{2}{c}{$N=16$} &
\multicolumn{2}{c}{$N=32$} &
\multicolumn{2}{c}{$N=64$} \\
&
(\%) &
p 	 &
(\%) &
p 	 &
(\%) &
p 	 &
(\%) &
p 	 &
(\%) &
p 	 &
(\%) &
p 	 \\ \midrule
\textbf{$M_s = 0.2$} & -5.5  & 0.89907 & -13.2 & 0.38327 & -8.7  & 0.54803 & 9.7   & 0.30792 & -3.2   & 0.66529 & -7.0    & 0.10096 \\
\textbf{$M_s = 0.3$} & -42.1 & 0.18044 & 1.9   & 0.93581 & -13.0 & 0.32875 & 4.4   & 0.75014 & 29.3   & 0.00649 & -5.0    & 0.29447 \\
\textbf{$M_s = 0.4$} & -20.9 & 0.55214 & 18.3  & 0.65124 & 16.8  & 0.32715 & -31.7 & 0.00172 & -42.0  & 0.00005 & -38.3   & 0.00000 \\
\textbf{$M_s = 0.5$} & 0.3   & 0.99318 & -16.4 & 0.29791 & -28.0 & 0.07229 & -39.1 & 0.00045 & -44.9  & 0.00016 & -31.4   & 0.00013 \\
\textbf{$M_s = 0.6$} & -79.2 & 0.01921 & -41.8 & 0.01744 & -40.2 & 0.00115 & -51.9 & 0.00000 & 132.2  & 0.00000 & 80.6    & 0.00000 \\
\textbf{$M_s = 0.7$} & -92.2 & 0.05186 & -28.6 & 0.19899 & -38.8 & 0.00318 & -40.4 & 0.00002 & 160.3  & 0.12355 & 6619.1  & 0.00000 \\
\textbf{$M_s = 0.8$} & -33.1 & 0.33714 & -37.7 & 0.01429 & -55.3 & 0.00008 & -48.0 & 0.00005 & -17.7  & 0.71927 & 16702.9 & 0.00000 \\
\textbf{$M_s = 0.9$} & 34.8  & 0.51138 & -53.0 & 0.00151 & -57.2 & 0.00001 & -58.9 & 0.00000 & -25.5  & 0.51399 & 6016.9  & 0.00000 \\
\textbf{$M_s = 1.0$} & -53.8 & 0.08247 & -33.3 & 0.06736 & -35.9 & 0.00585 & -40.2 & 0.00000 & -15.4  & 0.16010 & 12183.2 & 0.00000 \\
\textbf{$M_s = 1.1$} & 20.3  & 0.40907 & 8.7   & 0.53123 & 76.5  & 0.05059 & 60.5  & 0.01015 & 3416.3 & 0.00000 & 4296.4  & 0.00000 \\
\textbf{$M_s = 1.2$} & -40.0 & 0.13271 & 159.9 & 0.05210 & 50.1  & 0.29175 & 201.1 & 0.00000 & 1839.3 & 0.00000 & 1640.3  & 0.00000 \\ \hline
\end{tabular}
}
\end{table*}

%%%%%%%%%%%%%%%%%%%%%%%%%%%%%%%%%%%%%%%%%%%%%%%%%%%%%%%%%%%%%%%%%%%%%%%%
%%%%%%%%%%%%%%%%%%%%%%%%%%%%%%%%%%%%%%%%%%%%%%%%%%%%%%%%%%%%%%%%%%%%%%%%
%%%%%%%%%%%%%%%%%%%%%%%%%%%%%%%%%%%%%%%%%%%%%%%%%%%%%%%%%%%%%%%%%%%%%%%%
%%%%%%%%%%%%%%%%%%%%%%%%%%%%%%%%%%%%%%%%%%%%%%%%%%%%%%%%%%%%%%%%%%%%%%%%
%%%%%%%%%%%%%%%%%%%%%%%%%%%%%%%%%%%%%%%%%%%%%%%%%%%%%%%%%%%%%%%%%%%%%%%%

\begin{table*}[t!]
\centering
\caption{CIC values\\
Results (averaged over the 8 trials and the agents in the population) for increasing population size ($N \in [4, 64]$) and environments of decreasing difficulty (increasing $S_{eq} \propto M_s$).}
\label{tab: CIC values}
\resizebox{0.5\textwidth}{!}{%
\begin{tabular}{@{}r|ccccc@{}}
\toprule
 & \multicolumn{1}{c}{$N=4$} & \multicolumn{1}{c}{$N=8$} & \multicolumn{1}{c}{$N=16$} & \multicolumn{1}{c}{$N=32$} & \multicolumn{1}{c}{$N=64$} \\ \midrule
\textbf{$M_s = 0.2$} & 0.037 & 0.043 & 0.046 & 0.047 & 0.048 \\
\textbf{$M_s = 0.3$} & 0.037 & 0.043 & 0.047 & 0.050 & 0.054 \\
\textbf{$M_s = 0.4$} & 0.424 & 1.008 & 0.611 & 0.122 & 0.188 \\
\textbf{$M_s = 0.5$} & 0.908 & 0.841 & 0.542 & 0.225 & 0.144 \\
\textbf{$M_s = 0.6$} & 0.798 & 0.658 & 0.281 & 0.530 & 0.130 \\
\textbf{$M_s = 0.7$} & 0.710 & 0.710 & 0.454 & 0.292 & 0.210 \\
\textbf{$M_s = 0.8$} & 0.741 & 0.668 & 0.513 & 0.267 & 0.210 \\
\textbf{$M_s = 0.9$} & 0.465 & 0.429 & 0.378 & 0.251 & 0.266 \\
\textbf{$M_s = 1.0$} & 0.406 & 0.513 & 0.265 & 0.382 & 0.331 \\
\textbf{$M_s = 1.1$} & 0.141 & 0.180 & 0.135 & 0.221 & 0.296 \\
\textbf{$M_s = 1.2$} & 0.080 & 0.095 & 0.150 & 0.218 & 0.301 \\ \bottomrule
\end{tabular}%
}\bigskip\bigskip
\end{table*}
% \clearpage

%%%%%%%%%%%%%%%%%%%%%%%%%%%%%%%%%%%%%%%%%%%%%%%%%%%%%%%%%%%%%%%%%%%%%%%%
%%%%%%%%%%%%%%%%%%%%%%%%%%%%%%%%%%%%%%%%%%%%%%%%%%%%%%%%%%%%%%%%%%%%%%%%
%%%%%%%%%%%%%%%%%%%%%%%%%%%%%%%%%%%%%%%%%%%%%%%%%%%%%%%%%%%%%%%%%%%%%%%%
%%%%%%%%%%%%%%%%%%%%%%%%%%%%%%%%%%%%%%%%%%%%%%%%%%%%%%%%%%%%%%%%%%%%%%%%
%%%%%%%%%%%%%%%%%%%%%%%%%%%%%%%%%%%%%%%%%%%%%%%%%%%%%%%%%%%%%%%%%%%%%%%%

\begin{table*}[t!]
\centering
\caption{Social Welfare \\
Results (averaged over 8 trials) for varying signal size ($G = \{ 1, \frac{N}{2}, 23, N, 41, \frac{3N}{2} \}$, where $N=32$) and equilibrium stock multiplier ($M_s$ values of $0.7$, $0.8$ and $0.9$). The following results include:\\
(i) Absolute values,\\
(ii) Relative difference (\%), i.e., $(\text{Result}_{G=X} - \text{Result}_{G=1}) / \text{Result}_{G=1}$, where $\text{Result}_{G=X}$ denotes the achieved result using a signal of cardinality $X \in \{ \frac{N}{2}, 23, N, 41, \frac{3N}{2} \}$, and\\
(iii) Student's T-test p-values with respect to $G = 1$}
\label{tab: Varying Signal Size Social Welfare}
\resizebox{\textwidth}{!}{%
\begin{tabular}{@{}r|cccccc|ccccc|ccccc@{}}
\toprule
\multicolumn{1}{c|}{\textbf{}} &
  \multicolumn{6}{c|}{\textbf{Absolute Values}} &
  \multicolumn{5}{c|}{\textbf{Relative Difference (\%)}} &
  \multicolumn{5}{c}{\textbf{p-values}} \\
\multicolumn{1}{c|}{\textbf{}} &
  \textbf{$G = 1$} &
  \textbf{$G = \frac{N}{2}$} &
  \textbf{$G = 23$} &
  \textbf{$G = N$} &
  \textbf{$G = 41$} &
  \textbf{$G = \frac{3N}{2}$} &
  \textbf{$G = \frac{N}{2}$} &
  \textbf{$G = 23$} &
  \textbf{$G = N$} &
  \textbf{$G = 41$} &
  \textbf{$G = \frac{3N}{2}$} &
  \textbf{$G = \frac{N}{2}$} &
  \textbf{$G = 23$} &
  \textbf{$G = N$} &
  \textbf{$G = 41$} &
  \textbf{$G = \frac{3N}{2}$} \\ \midrule
\textbf{$M_s = 0.7$} & 34.14 & 184.82 & 158.95 & 196.57 & 219.33 & 221.43 & 441.3 & 365.6 & 475.7 & 542.4 & 548.5 & 0.00000 & 0.00014 & 0.00000 & 0.00000 & 0.00000 \\
\textbf{$M_s = 0.8$} & 86.95 & 249.51 & 260.94 & 316.08 & 376.67 & 423.03 & 187.0 & 200.1 & 263.5 & 333.2 & 386.5 & 0.00005 & 0.00004 & 0.00000 & 0.00000 & 0.00000 \\
\textbf{$M_s = 0.9$} & 165.75 & 431.43 & 497.05 & 637.30 & 784.72 & 970.48 & 160.3 & 199.9 & 284.5 & 373.4 & 485.5 & 0.00030 & 0.00007 & 0.00000 & 0.00000 & 0.00000 \\ \bottomrule
\end{tabular}%
}\bigskip\bigskip
\end{table*}

\begin{table*}[t!]
\centering
\caption{Episode Length (\#time-steps)\\
Results (averaged over 8 trials) for varying signal size ($G = \{ 1, \frac{N}{2}, 23, N, 41, \frac{3N}{2} \}$, where $N=32$) and equilibrium stock multiplier ($M_s$ values of $0.7$, $0.8$ and $0.9$). The following results include:\\
(i) Absolute values,\\
(ii) Relative difference (\%), i.e., $(\text{Result}_{G=X} - \text{Result}_{G=1}) / \text{Result}_{G=1}$, where $\text{Result}_{G=X}$ denotes the achieved result using a signal of cardinality $X \in \{ \frac{N}{2}, 23, N, 41, \frac{3N}{2} \}$, and\\
(iii) Student's T-test p-values with respect to $G = 1$}
\label{tab: Varying Signal Size Episode Length}
\resizebox{\textwidth}{!}{%
\begin{tabular}{@{}r|cccccc|ccccc|ccccc@{}}
\toprule
\multicolumn{1}{c|}{\textbf{}} &
  \multicolumn{6}{c|}{\textbf{Absolute Values}} &
  \multicolumn{5}{c|}{\textbf{Relative Difference (\%)}} &
  \multicolumn{5}{c}{\textbf{p-values}} \\
\multicolumn{1}{c|}{\textbf{}} &
  \textbf{$G = 1$} &
  \textbf{$G = \frac{N}{2}$} &
  \textbf{$G = 23$} &
  \textbf{$G = N$} &
  \textbf{$G = 41$} &
  \textbf{$G = \frac{3N}{2}$} &
  \textbf{$G = \frac{N}{2}$} &
  \textbf{$G = 23$} &
  \textbf{$G = N$} &
  \textbf{$G = 41$} &
  \textbf{$G = \frac{3N}{2}$} &
  \textbf{$G = \frac{N}{2}$} &
  \textbf{$G = 23$} &
  \textbf{$G = N$} &
  \textbf{$G = 41$} &
  \textbf{$G = \frac{3N}{2}$} \\ \midrule
\textbf{$M_s = 0.7$} & 70.78 & 500.00 & 440.16 & 500.00 & 500.00 & 500.00 & 606.4 & 521.9 & 606.4 & 606.4 & 606.4 & 0.00001 & 0.00072 & 0.00001 & 0.00001 & 0.00001 \\
\textbf{$M_s = 0.8$} & 200.63 & 500.00 & 500.00 & 500.00 & 500.00 & 500.00 & 149.2 & 149.2 & 149.2 & 149.2 & 149.2 & 0.00418 & 0.00418 & 0.00418 & 0.00418 & 0.00418 \\
\textbf{$M_s = 0.9$} & 275.00 & 500.00 & 500.00 & 500.00 & 500.00 & 500.00 & 81.8 & 81.8 & 81.8 & 81.8 & 81.8 & 0.01919 & 0.01919 & 0.01919 & 0.01919 & 0.01919 \\ \bottomrule
\end{tabular}%
}\bigskip\bigskip
\end{table*}

\begin{table*}[t!]
\centering
\caption{Training Time (\#episodes) \\
Results (averaged over 8 trials) for varying signal size ($G = \{ 1, \frac{N}{2}, 23, N, 41, \frac{3N}{2} \}$, where $N=32$) and equilibrium stock multiplier ($M_s$ values of $0.7$, $0.8$ and $0.9$). The following results include:\\
(i) Absolute values,\\
(ii) Relative difference (\%), i.e., $(\text{Result}_{G=X} - \text{Result}_{G=1}) / \text{Result}_{G=1}$, where $\text{Result}_{G=X}$ denotes the achieved result using a signal of cardinality $X \in \{ \frac{N}{2}, 23, N, 41, \frac{3N}{2} \}$, and\\
(iii) Student's T-test p-values with respect to $G = 1$}
\label{tab: Varying Signal Size Training Time}
\resizebox{\textwidth}{!}{%
\begin{tabular}{@{}r|cccccc|ccccc|ccccc@{}}
\toprule
\multicolumn{1}{c|}{\textbf{}} &
  \multicolumn{6}{c|}{\textbf{Absolute Values}} &
  \multicolumn{5}{c|}{\textbf{Relative Difference (\%)}} &
  \multicolumn{5}{c}{\textbf{p-values}} \\
\multicolumn{1}{c|}{\textbf{}} &
  \textbf{$G = 1$} &
  \textbf{$G = \frac{N}{2}$} &
  \textbf{$G = 23$} &
  \textbf{$G = N$} &
  \textbf{$G = 41$} &
  \textbf{$G = \frac{3N}{2}$} &
  \textbf{$G = \frac{N}{2}$} &
  \textbf{$G = 23$} &
  \textbf{$G = N$} &
  \textbf{$G = 41$} &
  \textbf{$G = \frac{3N}{2}$} &
  \textbf{$G = \frac{N}{2}$} &
  \textbf{$G = 23$} &
  \textbf{$G = N$} &
  \textbf{$G = 41$} &
  \textbf{$G = \frac{3N}{2}$} \\ \midrule
\textbf{$M_s = 0.7$} & 4757.00 & 4039.00 & 3493.00 & 3279.00 & 3147.00 & 3767.00 & -15.1 & -26.6 & -31.1 & -33.8 & -20.8 & 0.00001 & 0.01975 & 0.00001 & 0.00000 & 0.00001 \\
\textbf{$M_s = 0.8$} & 4875.75 & 3572.00 & 2933.00 & 2890.00 & 2680.00 & 3228.00 & -26.7 & -39.8 & -40.7 & -45.0 & -33.8 & 0.01919 & 0.00102 & 0.01919 & 0.00000 & 0.01919 \\
\textbf{$M_s = 0.9$} & 4511.13 & 3479.00 & 2386.00 & 2896.00 & 2575.00 & 2123.00 & -22.9 & -47.1 & -35.8 & -42.9 & -52.9 & 0.01919 & 0.00004 & 0.01919 & 0.00003 & 0.01919 \\ \bottomrule
\end{tabular}%
}\bigskip\bigskip
\end{table*}

\begin{table*}[t!]
\centering
\caption{Jain Index (higher is better) \\
Results (averaged over 8 trials) for varying signal size ($G = \{ 1, \frac{N}{2}, 23, N, 41, \frac{3N}{2} \}$, where $N=32$) and equilibrium stock multiplier ($M_s$ values of $0.7$, $0.8$ and $0.9$). The following results include:\\
(i) Absolute values,\\
(ii) Relative difference (\%), i.e., $(\text{Result}_{G=X} - \text{Result}_{G=1}) / \text{Result}_{G=1}$, where $\text{Result}_{G=X}$ denotes the achieved result using a signal of cardinality $X \in \{ \frac{N}{2}, 23, N, 41, \frac{3N}{2} \}$, and\\
(iii) Student's T-test p-values with respect to $G = 1$}
\label{tab: Varying Signal Size Jain Index}
\resizebox{\textwidth}{!}{%
\begin{tabular}{@{}r|cccccc|ccccc|ccccc@{}}
\toprule
\multicolumn{1}{c|}{\textbf{}} &
  \multicolumn{6}{c|}{\textbf{Absolute Values}} &
  \multicolumn{5}{c|}{\textbf{Relative Difference (\%)}} &
  \multicolumn{5}{c}{\textbf{p-values}} \\
\multicolumn{1}{c|}{\textbf{}} &
  \textbf{$G = 1$} &
  \textbf{$G = \frac{N}{2}$} &
  \textbf{$G = 23$} &
  \textbf{$G = N$} &
  \textbf{$G = 41$} &
  \textbf{$G = \frac{3N}{2}$} &
  \textbf{$G = \frac{N}{2}$} &
  \textbf{$G = 23$} &
  \textbf{$G = N$} &
  \textbf{$G = 41$} &
  \textbf{$G = \frac{3N}{2}$} &
  \textbf{$G = \frac{N}{2}$} &
  \textbf{$G = 23$} &
  \textbf{$G = N$} &
  \textbf{$G = 41$} &
  \textbf{$G = \frac{3N}{2}$} \\ \midrule
\textbf{$M_s = 0.7$} & 0.99635 & 0.99489 & 0.99594 & 0.99660 & 0.99744 & 0.99686 & -0.1 & 0.0 & 0.0 & 0.1 & 0.1 & 0.69758 & 0.91342 & 0.94743 & 0.77056 & 0.89174 \\
\textbf{$M_s = 0.8$} & 0.99023 & 0.99517 & 0.99665 & 0.99733 & 0.99742 & 0.99780 & 0.5 & 0.6 & 0.7 & 0.7 & 0.8 & 0.32098 & 0.20210 & 0.16118 & 0.15617 & 0.13698 \\
\textbf{$M_s = 0.9$} & 0.99234 & 0.99718 & 0.99777 & 0.99810 & 0.99863 & 0.99884 & 0.5 & 0.5 & 0.6 & 0.6 & 0.7 & 0.12907 & 0.09127 & 0.07503 & 0.05402 & 0.04728 \\ \bottomrule
\end{tabular}%
}\bigskip\bigskip
\end{table*}

\begin{table*}[t!]
\centering
\caption{Gini Coefficient \\
Results (averaged over 8 trials) for varying signal size ($G = \{ 1, \frac{N}{2}, 23, N, 41, \frac{3N}{2} \}$, where $N=32$) and equilibrium stock multiplier ($M_s$ values of $0.7$, $0.8$ and $0.9$). The following results include:\\
(i) Absolute values,\\
(ii) Relative difference (\%), i.e., $(\text{Result}_{G=X} - \text{Result}_{G=1}) / \text{Result}_{G=1}$, where $\text{Result}_{G=X}$ denotes the achieved result using a signal of cardinality $X \in \{ \frac{N}{2}, 23, N, 41, \frac{3N}{2} \}$, and\\
(iii) Student's T-test p-values with respect to $G = 1$}
\label{tab: Varying Signal Size Gini Coefficient}
\resizebox{\textwidth}{!}{%
\begin{tabular}{@{}r|cccccc|ccccc|ccccc@{}}
\toprule
\multicolumn{1}{c|}{\textbf{}} &
  \multicolumn{6}{c|}{\textbf{Absolute Values}} &
  \multicolumn{5}{c|}{\textbf{Relative Difference (\%)}} &
  \multicolumn{5}{c}{\textbf{p-values}} \\
\multicolumn{1}{c|}{\textbf{}} &
  \textbf{$G = 1$} &
  \textbf{$G = \frac{N}{2}$} &
  \textbf{$G = 23$} &
  \textbf{$G = N$} &
  \textbf{$G = 41$} &
  \textbf{$G = \frac{3N}{2}$} &
  \textbf{$G = \frac{N}{2}$} &
  \textbf{$G = 23$} &
  \textbf{$G = N$} &
  \textbf{$G = 41$} &
  \textbf{$G = \frac{3N}{2}$} &
  \textbf{$G = \frac{N}{2}$} &
  \textbf{$G = 23$} &
  \textbf{$G = N$} &
  \textbf{$G = 41$} &
  \textbf{$G = \frac{3N}{2}$} \\ \midrule
\textbf{$M_s = 0.7$} & 0.01252 & 0.04010 & 0.03525 & 0.03259 & 0.02801 & 0.03110 & 220.3 & 181.5 & 160.3 & 123.7 & 148.4 & 0.04090 & 0.08751 & 0.12355 & 0.22435 & 0.15164 \\
\textbf{$M_s = 0.8$} & 0.03503 & 0.03873 & 0.03253 & 0.02884 & 0.02838 & 0.02590 & 10.6 & -7.1 & -17.7 & -19.0 & -26.1 & 0.82971 & 0.88448 & 0.71927 & 0.69934 & 0.59813 \\
\textbf{$M_s = 0.9$} & 0.03275 & 0.02992 & 0.02646 & 0.02440 & 0.02065 & 0.01885 & -8.7 & -19.2 & -25.5 & -37.0 & -42.4 & 0.82380 & 0.62158 & 0.51399 & 0.34755 & 0.28352 \\ \bottomrule
\end{tabular}%
}\bigskip\bigskip\bigskip\bigskip\bigskip\bigskip\bigskip\bigskip
\end{table*}

%%%%%%%%%%%%%%%%%%%%%%%%%%%%%%%%%%%%%%%%%%%%%%%%%%%%%%%%%%%%%%%%%%%%%%%%
%%%%%%%%%%%%%%%%%%%%%%%%%%%%%%%%%%%%%%%%%%%%%%%%%%%%%%%%%%%%%%%%%%%%%%%%
%%%%%%%%%%%%%%%%%%%%%%%%%%%%%%%%%%%%%%%%%%%%%%%%%%%%%%%%%%%%%%%%%%%%%%%%
%%%%%%%%%%%%%%%%%%%%%%%%%%%%%%%%%%%%%%%%%%%%%%%%%%%%%%%%%%%%%%%%%%%%%%%%

\begin{table*}[t!]
\centering
\caption{Social Welfare, Episode Length, Training Time, Jain Index, Gini Coefficient\\
Results (averaged over 8 trials) for higher growth rate ($r = 2$), with ($G=N$) and without ($G=1$) the introduced signal, for environments of decreasing difficulty (increasing $S_{eq} \propto M_s$), and population size $N = 64$.}
\label{tab: Growth Rate 2}
\resizebox{\textwidth}{!}{%
\begin{tabular}{@{}r|cccccccccc@{}}
\toprule
\textbf{} &
  \multicolumn{2}{c}{\textbf{Social Welfare}} &
  \multicolumn{2}{c}{\textbf{Episode Length}} &
  \multicolumn{2}{c}{\textbf{Training Time}} &
  \multicolumn{2}{c}{\textbf{Jain Index}} &
  \multicolumn{2}{c}{\textbf{Gini Coefficient}} \\
\textbf{} &
  \textbf{$G = 1$} &
  \textbf{$G = N$} &
  \textbf{$G = 1$} &
  \textbf{$G = N$} &
  \textbf{$G = 1$} &
  \textbf{$G = N$} &
  \textbf{$G = 1$} &
  \textbf{$G = N$} &
  \textbf{$G = 1$} &
  \textbf{$G = N$} \\ \midrule
\textbf{$M_s = 0.2$} & 7.52     & 7.53     & 1.04   & 1.05   & 5000.00 & 5000.00 & 0.99955 & 0.99959 & 0.01193 & 0.01129 \\
\textbf{$M_s = 0.3$} & 11.65    & 11.57    & 1.15   & 1.13   & 5000.00 & 5000.00 & 0.99953 & 0.99963 & 0.01215 & 0.01084 \\
\textbf{$M_s = 0.4$} & 21.00    & 19.89    & 2.54   & 2.24   & 5000.00 & 5000.00 & 0.99776 & 0.99795 & 0.02611 & 0.02549 \\
\textbf{$M_s = 0.5$} & 23.52    & 22.52    & 2.52   & 2.19   & 5000.00 & 5000.00 & 0.99467 & 0.99724 & 0.04075 & 0.02941 \\
\textbf{$M_s = 0.6$} & 23.91    & 26.33    & 1.71   & 2.72   & 5000.00 & 5000.00 & 0.99243 & 0.99696 & 0.04894 & 0.03108 \\
\textbf{$M_s = 0.7$} & 26.62    & 354.47   & 3.08   & 306.34 & 5000.00 & 5000.00 & 0.99066 & 0.99728 & 0.05172 & 0.02909 \\
\textbf{$M_s = 0.8$} & 219.16   & 922.99   & 190.03 & 453.91 & 4957.13 & 5000.00 & 0.97867 & 0.99834 & 0.06876 & 0.02285 \\
\textbf{$M_s = 0.9$} & 1320.48  & 7171.94  & 500.00 & 500.00 & 4900.00 & 2053.00 & 0.97524 & 0.99911 & 0.07791 & 0.01666 \\
\textbf{$M_s = 1.0$} & 2490.16  & 16849.49 & 500.00 & 500.00 & 3003.00 & 909.00  & 0.99858 & 0.99890 & 0.00612 & 0.01856 \\
\textbf{$M_s = 1.1$} & 17286.26 & 18737.69 & 500.00 & 500.00 & 853.00  & 1193.00 & 0.99997 & 0.99914 & 0.00192 & 0.01640 \\
\textbf{$M_s = 1.2$} & 20740.25 & 19354.80 & 500.00 & 500.00 & 781.00  & 1427.00 & 0.99996 & 0.99930 & 0.00220 & 0.01480 \\ \bottomrule
\end{tabular}%
}\bigskip\bigskip
\end{table*}

\begin{table*}[t!]
\centering
\caption{Social Welfare, Episode Length, Training Time, Jain Index, Gini Coefficient\\
(i) Relative difference in the achieved result when signal of cardinality $G=N$ is introduced ($(\text{Result}_{G=N} - \text{Result}_{G=1}) / \text{Result}_{G=1}$, where $\text{Result}_{G=X}$ denotes the achieved result using a signal of cardinality $X$), and \\
(ii) Student's T-test p-values, \\
Results (averaged over 8 trials) for higher growth rate ($r = 2$), with ($G=N$) and without ($G=1$) the introduced signal, for environments of decreasing difficulty (increasing $S_{eq} \propto M_s$), and population size $N = 64$.}
\label{tab: Growth Rate 2 relative difference and p values}
\resizebox{\textwidth}{!}{%
\begin{tabular}{@{}r|cccccccccc@{}}
\toprule
\textbf{} &
  \multicolumn{2}{c}{\textbf{Social Welfare}} &
  \multicolumn{2}{c}{\textbf{Episode Length}} &
  \multicolumn{2}{c}{\textbf{Training Time}} &
  \multicolumn{2}{c}{\textbf{Jain Index}} &
  \multicolumn{2}{c}{\textbf{Gini Coefficient}} \\
\textbf{} &
  \textbf{(\%)} &
  \textbf{p-value} &
  \textbf{(\%)} &
  \textbf{p-value} &
  \textbf{(\%)} &
  \textbf{p-value} &
  \textbf{(\%)} &
  \textbf{p-value} &
  \textbf{(\%)} &
  \textbf{p-value} \\ \midrule
\textbf{$M_s = 0.2$} & 0.1    & 0.19975 & 0.2    & 0.17555 & 0.0   & NaN     & 0.0  & 0.21075 & -5.4  & 0.20205 \\
\textbf{$M_s = 0.3$} & -0.7   & 0.00261 & -1.7   & 0.00274 & 0.0   & NaN     & 0.0  & 0.00050 & -10.8 & 0.00096 \\
\textbf{$M_s = 0.4$} & -5.3   & 0.66393 & -11.8  & 0.63772 & 0.0   & NaN     & 0.0  & 0.41637 & -2.4  & 0.65179 \\
\textbf{$M_s = 0.5$} & -4.2   & 0.52530 & -13.1  & 0.53060 & 0.0   & NaN     & 0.3  & 0.00004 & -27.8 & 0.00001 \\
\textbf{$M_s = 0.6$} & 10.1   & 0.01996 & 59.3   & 0.02644 & 0.0   & NaN     & 0.5  & 0.00000 & -36.5 & 0.00000 \\
\textbf{$M_s = 0.7$} & 1231.4 & 0.00000 & 9862.3 & 0.00000 & 0.0   & NaN     & 0.7  & 0.00005 & -43.8 & 0.00001 \\
\textbf{$M_s = 0.8$} & 321.1  & 0.00006 & 138.9  & 0.01041 & 0.9   & 0.33428 & 2.0  & 0.00097 & -66.8 & 0.00204 \\
\textbf{$M_s = 0.9$} & 443.1  & 0.00000 & 0.0    & NaN     & -58.1 & 0.00000 & 2.4  & 0.00000 & -78.6 & 0.00000 \\
\textbf{$M_s = 1.0$} & 576.6  & 0.00000 & 0.0    & NaN     & -69.7 & 0.00000 & 0.0  & 0.51998 & 203.5 & 0.00000 \\
\textbf{$M_s = 1.1$} & 8.4    & 0.00000 & 0.0    & NaN     & 39.9  & 0.00117 & -0.1 & 0.00000 & 753.5 & 0.00000 \\
\textbf{$M_s = 1.2$} & -6.7   & 0.00000 & 0.0    & NaN     & 82.7  & 0.00000 & -0.1 & 0.00000 & 573.6 & 0.00000 \\ \bottomrule
\end{tabular}%
}\bigskip\bigskip
\end{table*}

%%%%%%%%%%%%%%%%%%%%%%%%%%%%%%%%%%%%%%%%%%%%%%%%%%%%%%%%%%%%%%%%%%%%%%%%
%%%%%%%%%%%%%%%%%%%%%%%%%%%%%%%%%%%%%%%%%%%%%%%%%%%%%%%%%%%%%%%%%%%%%%%%
%%%%%%%%%%%%%%%%%%%%%%%%%%%%%%%%%%%%%%%%%%%%%%%%%%%%%%%%%%%%%%%%%%%%%%%%
%%%%%%%%%%%%%%%%%%%%%%%%%%%%%%%%%%%%%%%%%%%%%%%%%%%%%%%%%%%%%%%%%%%%%%%%

\begin{table*}[t!]
\centering
\caption{Average number of agents in each bin (i.e., harvesting with effort $\epsilon \in [0 - 0.33)$ (`idle'), $[0.33 - 0.66)$ (`moderate'), and $[0.66 - 1]$ (`active')). The presented values start from the first equilibrium stock multiplier ($M_s$) where a non-depleting strategy was achieved in each setting.\\
Results (averaged over 8 trials) for increasing population size, with ($G=N$) and without ($G=1$) the introduced signal, for environments of decreasing difficulty (increasing $S_{eq} \propto M_s$).}
\label{tab: Access Rate}

\resizebox{0.8\textwidth}{!}{%
\begin{tabular}{@{}r|cccccccccc@{}}
\multicolumn{11}{l}{\textbf{Number of `idle' agents: $\epsilon \in [0 - 0.33)$}} \\
\toprule
\textbf{} &
  \multicolumn{2}{c}{\textbf{$N = 8$}} &
  \multicolumn{2}{c}{\textbf{$N = 16$}} &
  \multicolumn{2}{c}{\textbf{$N = 32$}} &
  \multicolumn{2}{c}{\textbf{$N = 64$, $r = 1$}} &
  \multicolumn{2}{c}{\textbf{$N = 64$, $r = 2$}} \\
\textbf{} &
  \textbf{$G = 1$} &
  \textbf{$G = N$} &
  \textbf{$G = 1$} &
  \textbf{$G = N$} &
  \textbf{$G = 1$} &
  \textbf{$G = N$} &
  \textbf{$G = 1$} &
  \textbf{$G = N$} &
  \textbf{$G = 1$} &
  \textbf{$G = N$} \\ \midrule
\textbf{$M_s = 0.2$} &  &  &  &  &  &  &  &  &  &  \\
\textbf{$M_s = 0.3$} &  &  &  &  &  &  &  &  &  &  \\
\textbf{$M_s = 0.4$} &  & \textbf{4.9} &  & \textbf{9.4} &  &  &  &  &  &  \\
\textbf{$M_s = 0.5$} & \textbf{0.4} & 3.7 &  & 6.9 &  &  &  &  &  &  \\
\textbf{$M_s = 0.6$} & 0.0 & 2.8 & \textbf{0.0} & 5.0 &  & \textbf{12.1} &  &  &  &  \\
\textbf{$M_s = 0.7$} & 0.0 & 1.8 & 0.0 & 3.3 &  & 6.6 &  &  &  & \textbf{31.6} \\
\textbf{$M_s = 0.8$} & 0.0 & 1.0 & 0.0 & 2.0 &  & 3.6 &  &  &  & 5.9 \\
\textbf{$M_s = 0.9$} & 0.0 & 0.4 & 0.0 & 0.8 &  & 1.4 &  & \textbf{2.0} &  & 1.8 \\
\textbf{$M_s = 1.0$} & 0.0 & 0.2 & 0.0 & 0.2 & \textbf{0.0} & 0.3 & \textbf{0.0} & 1.2 & \textbf{0.0} & 1.6 \\
\textbf{$M_s = 1.1$} & 0.0 & 0.0 & 0.0 & 0.0 & 0.0 & 0.2 & 0.0 & 1.6 & 0.0 & 1.3 \\
\textbf{$M_s = 1.2$} & 0.0 & 0.0 & 0.0 & 0.0 & 0.0 & 0.2 & 0.0 & 1.9 & 0.0 & 1.1 \\ \bottomrule
\end{tabular}%
} \bigskip\bigskip

\resizebox{0.8\textwidth}{!}{%
\begin{tabular}{@{}r|cccccccccc@{}}
\multicolumn{11}{l}{\textbf{Number of `moderate' agents: $\epsilon \in [0.33 - 0.66)$}} \\
\toprule
\textbf{} &
  \multicolumn{2}{c}{\textbf{$N = 8$}} &
  \multicolumn{2}{c}{\textbf{$N = 16$}} &
  \multicolumn{2}{c}{\textbf{$N = 32$}} &
  \multicolumn{2}{c}{\textbf{$N = 64$, $r = 1$}} &
  \multicolumn{2}{c}{\textbf{$N = 64$, $r = 2$}} \\
\textbf{} &
  \textbf{$G = 1$} &
  \textbf{$G = N$} &
  \textbf{$G = 1$} &
  \textbf{$G = N$} &
  \textbf{$G = 1$} &
  \textbf{$G = N$} &
  \textbf{$G = 1$} &
  \textbf{$G = N$} &
  \textbf{$G = 1$} &
  \textbf{$G = N$} \\ \midrule
\textbf{$M_s = 0.2$} &  &  &  &  &  &  &  &  &  &  \\
\textbf{$M_s = 0.3$} &  &  &  &  &  &  &  &  &  &  \\
\textbf{$M_s = 0.4$} &  & \textbf{0.7} &  & \textbf{2.0} &  &  &  &  &  &  \\
\textbf{$M_s = 0.5$} & \textbf{6.8} & 1.0 &  & 2.5 &  &  &  &  &  &  \\
\textbf{$M_s = 0.6$} & 6.4 & 1.3 & \textbf{13.4} & 3.0 &  & \textbf{5.8} &  &  &  &  \\
\textbf{$M_s = 0.7$} & 3.6 & 1.5 & 6.1 & 3.0 &  & 5.6 &  &  &  & \textbf{6.6} \\
\textbf{$M_s = 0.8$} & 1.6 & 1.5 & 1.5 & 2.3 &  & 4.8 &  &  &  & 11.6 \\
\textbf{$M_s = 0.9$} & 0.4 & 0.9 & 0.6 & 1.6 &  & 3.3 &  & \textbf{7.8} &  & 7.8 \\
\textbf{$M_s = 1.0$} & 0.1 & 0.5 & 0.1 & 0.5 & \textbf{0.4} & 1.5 & \textbf{0.0} & 6.3 & \textbf{0.1} & 6.6 \\
\textbf{$M_s = 1.1$} & 0.0 & 0.0 & 0.0 & 0.0 & 0.0 & 0.9 & 0.0 & 6.4 & 0.0 & 5.9 \\
\textbf{$M_s = 1.2$} & 0.0 & 0.0 & 0.0 & 0.0 & 0.0 & 1.4 & 0.0 & 7.0 & 0.0 & 5.7 \\ \bottomrule
\end{tabular}%
} \bigskip\bigskip

\resizebox{0.8\textwidth}{!}{%
\begin{tabular}{@{}r|cccccccccc@{}}
\multicolumn{11}{l}{\textbf{Number of `active' agents: $\epsilon \in [0.66 - 1]$}} \\
\toprule
\textbf{} &
  \multicolumn{2}{c}{\textbf{$N = 8$}} &
  \multicolumn{2}{c}{\textbf{$N = 16$}} &
  \multicolumn{2}{c}{\textbf{$N = 32$}} &
  \multicolumn{2}{c}{\textbf{$N = 64$, $r = 1$}} &
  \multicolumn{2}{c}{\textbf{$N = 64$, $r = 2$}} \\
\textbf{} &
  \textbf{$G = 1$} &
  \textbf{$G = N$} &
  \textbf{$G = 1$} &
  \textbf{$G = N$} &
  \textbf{$G = 1$} &
  \textbf{$G = N$} &
  \textbf{$G = 1$} &
  \textbf{$G = N$} &
  \textbf{$G = 1$} &
  \textbf{$G = N$} \\ \midrule
\textbf{$M_s = 0.2$} &  &  &  &  &  &  &  &  &  &  \\
\textbf{$M_s = 0.3$} &  &  &  &  &  &  &  &  &  &  \\
\textbf{$M_s = 0.4$} &  & \textbf{2.4} &  & \textbf{4.5} &  &  &  &  &  &  \\
\textbf{$M_s = 0.5$} & \textbf{0.9} & 3.3 &  & 6.6 &  &  &  &  &  &  \\
\textbf{$M_s = 0.6$} & 1.6 & 3.9 & \textbf{2.6} & 8.0 &  & \textbf{14.1} &  &  &  &  \\
\textbf{$M_s = 0.7$} & 4.4 & 4.7 & 9.9 & 9.8 &  & 19.8 &  &  &  & \textbf{25.7} \\
\textbf{$M_s = 0.8$} & 6.4 & 5.5 & 14.5 & 11.7 &  & 23.6 &  &  &  & 46.6 \\
\textbf{$M_s = 0.9$} & 7.6 & 6.6 & 15.4 & 13.6 &  & 27.3 &  & \textbf{54.3} &  & 54.4 \\
\textbf{$M_s = 1.0$} & 7.9 & 7.4 & 15.9 & 15.3 & \textbf{31.6} & 30.2 & \textbf{64.0} & 56.5 & \textbf{63.9} & 55.8 \\
\textbf{$M_s = 1.1$} & 8.0 & 8.0 & 16.0 & 16.0 & 32.0 & 30.9 & 64.0 & 56.0 & 64.0 & 56.8 \\
\textbf{$M_s = 1.2$} & 8.0 & 8.0 & 16.0 & 16.0 & 32.0 & 30.4 & 64.0 & 55.2 & 64.0 & 57.2 \\ \bottomrule
\end{tabular}%
}
\end{table*}

\end{document}